%% file: badges.tex
\documentclass[11pt]{amsart}
\usepackage[margin=1.2in]{geometry}
\usepackage{graphicx}
\usepackage{amssymb}
\usepackage{epstopdf}
\usepackage{subfigure}

\usepackage[]{color-edits}
\addauthor{vs}{red}
\addauthor{gs}{blue}
\addauthor{ni}{green}

\newtheorem{theorem}{Theorem}
\newtheorem{lemma}[theorem]{Lemma}

\newtheorem{proposition}[theorem]{Proposition}

\newtheorem{definition}[theorem]{Definition}

\def\squareforqed{\hbox{\rule{2.5mm}{2.5mm}}}

\def\QED{\ifmmode\squareforqed 
  \else{\nobreak\hfil   
    \penalty50                 
    \hskip1em                  
    \null                      
    \nobreak                   
    \hfil                      
    \squareforqed              
    \parfillskip=0pt           
    \finalhyphendemerits=0     
    \endgraf}                  
  \fi}

\def\blksquare{\rule{2mm}{2mm}}
\def\qedsymbol{\blksquare}
\newcommand{\bg}[1]{\medskip\noindent{\bf #1}}
\newcommand{\ed}{{\hfill\qedsymbol}\medskip}

\newenvironment{proofof}[1]{{\it{Proof of #1 : }}}{\ed}

\newenvironment{example}{\bg{Example. }}{\ed}

\newcommand{\E}{\ensuremath{\textsc{E}}}
\newcommand{\R}{\ensuremath{\mathbb{R}}}

\newcommand{\stat}[1]{\ensuremath{S\left(#1\right)}}
\newcommand{\statn}[2]{\ensuremath{S_{#1}\left(#2\right)}}
\newcommand{\statnpr}[2]{\ensuremath{S_{#1}'\left(#2\right)}}
\newcommand{\cdf}[1][]{\ensuremath{\ifthenelse{\equal{#1}{}}{F}{F\left(#1\right)}}}
\newcommand{\pdf}[1][]{\ensuremath{\ifthenelse{\equal{#1}{}}{f}{f\left(#1\right)}}}
\newcommand{\cdfi}[1][]{\ensuremath{\ifthenelse{\equal{#1}{}}{F^{-1}}{F^{-1}\left(#1\right)}}}
\newcommand{\xq}[1]{\ensuremath{\hat{x}\left(#1\right)}}

\newcommand{\xqpr}[1]{\ensuremath{\hat{x}'\left(#1\right)}}
\newcommand{\val}[1][]{\ensuremath{v_{#1}}}
\newcommand{\quant}[1][]{\ensuremath{q_{#1}}}
\newcommand{\tb}[1]{\ensuremath{\theta_{#1}}}
\newcommand{\tbvec}{\ensuremath{{\boldsymbol\theta}}}
\newcommand{\tv}[1]{\ensuremath{a_{#1}}}

\newcommand{\tq}[1]{\ensuremath{\kappa_{#1}}}
\newcommand{\tqvec}{\ensuremath{{\boldsymbol \kappa}}}
\newcommand{\tqstar}{\ensuremath{\kappa^*}}

\newcommand{\acdf}[2]{\ensuremath{\ifthenelse{\equal{#2}{}}{F_{#1}}{F_{#1}\left(#2\right)}}}

\newcommand{\bvec}{\ensuremath{{\bf b}}}
\newcommand{\rank}{\ensuremath{r}}
\newcommand{\rvec}{\ensuremath{{\bf r}}}
\newcommand{\qvec}{\ensuremath{{\bf q}}}

\newcommand{\opt}{\ensuremath{\textsc{Opt}}}
\newcommand{\apx}{\ensuremath{\textsc{Apx}}}

\title{Social Status and Badge Design}

\author{Nicole Immorlica \and Greg Stoddard \and Vasilis Syrgkanis}

\begin{document}

\begin{abstract}
Many websites rely on user-generated content to provide value to consumers.  These websites typically incentivize participation by awarding users badges based on their contributions.  While these badges typically have no explicit value, they act as  symbols of social status within a community. In this paper, we consider the design of badge mechanisms for the objective of maximizing the total contributions made to a website. Users exert costly effort to make contributions and, in return, are awarded with badges. A badge is only valued to the extent that it signals social status and thus badge valuations are determined endogenously by the number of users who earn each badge. The goal of this paper is to study the design of optimal and approximately badge mechanisms under these status valuations. We characterize badge mechanisms by whether they use a coarse partitioning scheme, i.e. awarding the same badge to many users, or use a fine partitioning scheme, i.e. awarding a unique badge to most users. We find that the optimal mechanism uses both fine partitioning and coarse partitioning. When status valuations exhibit a decreasing marginal value property, we prove that coarse partitioning is a necessary feature of any approximately optimal mechanism. Conversely, when status valuations exhibit an increasing marginal value property, we prove that fine partitioning is necessary for approximate optimality. 
\end{abstract}

\maketitle

\section{Introduction}\label{sec:intro}
\input{intro.tex}

\section{Related Work}\label{sec:related}
\input{relatedWork.tex}

\section{Badge Mechanisms and Social Status}\label{sec:model}
\input{model.tex}

\section{Preliminaries and Connection to Optimal Auction Design}\label{sec:preliminaries}
\input{preliminaries.tex}

\section{Optimal Badge Mechanism}\label{sec:optimal}
\input{optimal}

\section{Absolute Threshold Mechanisms}\label{sec:single}
\input{single}

\paragraph{Lower Bounds}\label{sec:lowerbound}
\input{lowerBound}
\input{many}


\section{Approximation with Leaderboards}\label{sec:complete-ranking}
\input{complete-ranking}

\section{Robustness to Treatment of Equal Status Opponents}\label{sec:tie-breaking}
\input{tie-breaking}

\section{Conclusion}
\input{conclusion}


\bibliographystyle{plain}
\bibliography{references}

\appendix
\section{Appendix}
\subsection{Proofs from Section \ref{sec:optimal}}
\input{appendix-opt}
\subsection{Proofs from Section \ref{sec:single}}
\input{appendix-equilibrium}

\input{appendix-single}

\subsection{Approximation with Many Absolute Badges}\label{sec:many-badges}
\input{appendix-many}

\subsection{Proofs from Section \ref{sec:complete-ranking}}
\input{appendix-complete-ranking}
\subsection{Proof from Section \ref{sec:tie-breaking}}\label{sec:appendix-tie-breaking}
\input{appendix-tie-breaking}


\end{document}

%% file: intro.tex
A number of popular websites are driven by user-generated content. Review sites such as Yelp and TripAdvisor need users to rate and review restaurants and hotels, social news aggregators like Reddit rely on users to submit and vote upon articles from around the web, and question and answer sites like Stack Overflow and Quora depend on their users to ask good questions and provide good answers. One threat to the success of such sites is the free-rider problem, that not enough users will the expend the effort to meaningfully contribute. Stack Overflow, a very successful user-driven Q\&A site for programming questions, addresses this free-riding problem with a system of \emph{reputation points} and \emph{badges}. A badge is a small symbol displayed on a user's profile and is awarded for accomplishing tasks like asking ten questions or providing the best answer to a given question. 


While badge systems have existed for a long time in the form of military medals, boy scout badges, customer status programs, etc, they have become particularly popular in web communities over the past few years. The Huffington Post recently implemented a badge system to reward actions from their commenters; these badges are now prominently displayed next to usernames in the comment sections. Amazon awards a badge for being in top 1000 reviewers, another for being in the top 100, etc. The Mozilla Foundation is leading an initiative called the Open Badge Project which hopes to set an open standard for awarding, collecting, and displaying badges across any platform throughout the web.
 Their ultimate goal is to provide a persistent collection of badges that can be displayed anywhere on the web as a proof of acquired skills and achievements.

With all this excitement and energy surrounding badge systems, one fundamental question emerges: why do people care about badges? On the surface level, a badge is just a small sticker or group of pixels on a profile, so why should it be that badges incentivize people to exert effort to earn them? This is the very reason that makes such systems attractive to platforms; they are getting ``something for nothing'' by awarding badges which cost them nothing, and in return they receive increased contributions from their users. Despite the lack of an explicit value, it is quite clear that people are motivated by badges \cite{anderson2013steering}. In this work, we take the view that badges (or any virtual reward) do not have fixed values but that their value is endogenously determined by how many others have earned the same badge (or a better one). The most valuable badges not only signal that a user has accomplished some impressive task but also signal that relatively few users have earned them. Top level badges, such as the ``gold'' badges on Stack Overflow or the ``Superuser'' badge on Huffington Post, are described as difficult to earn and only awarded to the most committed users. As more users earn a particular badge, that badge loses its ability to distinguish a member within the community and thus becomes less valuable. This work studies optimal badge design when badges are a means for establishing \emph{social status}. 

While it has long been accepted that people are motivated and affected by concerns of social status, the exact role and nature of status remains an issue of contention. In particular, there's a debate over whether status is an end in itself, or whether it serves as a useful proxy for some other objective. On the web, a variety of status concerns arise from the idiosyncrasies of different communities. On Wikipedia, there's evidence that high status editors enjoy a higher chance of having their suggestions adopted by administrators \cite{keegan2010egalitarians}. 
The Huffington Post gives increased attention to comments made by users who have earned the ``pundit'' badge by placing those comments at the top of each discussion. Top reviewers on Amazon receive large amount of free products from manufacturers. It could also be that users simply value having high status because it makes them feel good, or signals something about their abilities or contributions to others. Rather than study the provenance of these concerns across different platforms, we seek to understand how different properties of status concerns should influence the design of badge mechanisms.

The goal of this paper is to examine the incentives created by social status concerns and to study the design of optimal social-psychological incentive schemes. In practice, there is a wide variety of such mechanisms such as awarding digital badges for accomplishments, assigning titles to users, maintaining a leaderboard where users are ranked by accomplishments, etc. The key feature of these schemes is that they induce a partitioning of users into an ordered set of status classes; users that earn the no reward are in the lowest status class while users that earn the highest reward are in the top status class. We abstract away the particular details of rewards and focus on the mechanism's role in assigning users into these status classes. Throughout this paper, we generically refer to a \emph{badge} as the assigned status class of a user. 

The high level question we study is how should a designer award badges to maximize the amount of content contributed by the users of a website. We answer this question by using a game-theoretic model where users strategically choose the amount they contribute to maximize their utility, which is defined as the status value they receive minus their cost of producing that content.Our goal is to give a broad characterization of the design of badge mechanisms with social status concerns. Our work addresses these the three main questions:


\begin{enumerate}

\item What is the badge mechanism that maximizes the expected amount of content contributed by users? 

\item What can be done with simple badge mechanisms? In some settings, the optimal mechanism may not be possible to implement. Perhaps the optimal mechanism is more complex than desired or imposes too much of an information requirement on the part of the designer. Motivated by these constraints, we study the performance of two variants of commonly-used badge mechanisms. 

\item How does the nature of status valuations affect the design of good mechanisms? Rather than assume a specific form of status valuations, we allow for users to have some general status value function and examine how the properties of this function influence the (approximate) optimality of the badge mechanisms we study. 

\end{enumerate}

\subsection{Our Contributions}

We model and analyze a game of incomplete information where users simultaneously make costly contributions to a site. Each user has some privately-known ability that determines their cost of contributing to the site. Based on their contributions, each users is assigned a single \emph{badge} out of an ordered set of badges. Users only derive value from these badges because of the \emph{status} they confer, where a user's status is defined by the fraction of users that have earned an equal or better badge. A user's value for this status is given by some function $\stat{\cdot}$. 

We study the design of optimal and approximately mechanisms in this setting. We define a badge mechanism as a function that awards badges to users based on their contribution and the contribution of all other users. The objective of the designer is to maximize the expected sum of contributions received from all users. We note that we can extend this framework to incorporate objectives such as as maximizing contributions that are at least of a certain quality (to mitigate agents submitting large amounts of low quality content) by redefining what constitutes an acceptable contribution. We analyze this optimization problem by using the well-established connection between crowdsourcing contests and all-pay auctions and tools from optimal mechanism design theory. We show that the status value functions ends up playing the role of an allocation function in a standard mechanism design setting. However, relative to a standard mechanism design setting, there are significant constraints on the design problem. The endogenous nature of status valuations introduces negative externalities between users; as more users earn a particular badge, the status value of that badge decreases. These negative externalities constrain the set of feasible status allocations implementable in equilibrium. 


We first prove that the optimal mechanism is a leaderboard with a cutoff. This mechanism assigns the lowest badge to any user whose contribution falls below a certain threshold, and assigns unique badges to the remaining users in decreasing order of their contribution. The optimality of this mechanism is not sensitive to the status value function $\stat{\cdot}$. On the one hand, this is a nice property for the optimal mechanism to have. On the other hand, studying the optimal mechanism does not elucidate the role of the status function. Indeed, we show that this insensitivity to the status value function does not extend to approximately optimal badge mechanisms. We examine the effects of the convexity of the status value function through two commonly-used types of badge mechanisms. We find that the ``shape'' of status valuations, i.e. whether $\stat{\cdot}$ is concave or convex with respect to status, plays a large role in the design of approximately optimal mechanisms.  

The first of these type of mechanisms is an absolute threshold mechanism, which is a mechanism defined by a vector of contribution thresholds such that the badge each user earns is determined by the maximum threshold that his contribution exceeds. We prove that an absolute threshold mechanism with a single threshold achieves a 4-approximation to the optimal mechanism for concave and linear status functions. 
For strictly convex status valuations, no absolute threshold mechanism can achieve a finite approximation with any constant number of badges. This indicates that any approximately optimal absolute threshold mechanism must increase the number of badges it uses as the user population increases but we show that the necessary number of badges to achieve a constant approximation is only logarithmic in a natural parameter of the status value function. 

We finally study the leaderboard mechanism with no cut-off, i.e. the mechanism that assigns unique badges to all users in decreasing order of their contribution. We prove the leaderboard mechanism is a 2-approximation to the optimal mechanism for any convex value function. However, this mechanism cannot achieve any finite approximation for the class of concave valuations. Contrasting this result with the optimal mechanism demonstrates the necessity of having some threshold below which any user earns the lowest badge. 

When viewed all together, these results give an important intuition regarding the partitioning of users. The function of a mechanism is to assign badges to users which partitions users into different sets of ordered status classes. To use the language of \cite{moldovanu2007contests}, a fine partitioning of users is one where users (or most of them) are assigned a unique badge, while a coarse partitioning of users is one where many users are assigned to the same badge. The optimal mechanism, the leaderboard with a cutoff, uses both types of partitioning; all users above the contribution threshold are finely partitioned while all users below the threshold are coarsely partitioned.

For the class of convex valuations, fine partitioning is necessary to achieve any reasonable approximation. With convex valuations, the marginal value of increasing a user's status increases as his status improves, so users are naturally motivated to climb as high as possible. Any approximately optimal mechanism must have enough badges to allow users to distinguish themselves. Both the optimal mechanism and the the leaderboard mechanism have this property but the absolute threshold mechanism with a single threshold does not. The partitioning in this latter mechanism is too coarse to achieve any good approximation   

For the class of concave valuations, coarse partitioning among the lower-ability users is necessary for approximate optimality. With concave valuations, the marginal status value diminishes as users improve their status and hence the natural motivations for users is to avoid being at the bottom of the population but not to care too much for rising further. In extreme instances, most of a user's marginal status value is gained from moving up from the lowest status. In these settings, a mechanism must coarsely partition all low-ability users into a single badge to provide sufficient incentive for high-ability users to earn the next highest badge. The leaderboard with a cut-off and the single absolute threshold mechanism both implement this coarse partitioning of low ability users but the leaderboard mechanism partitions too finely to achieve any good approximation.

In our last section, we extend our analysis to a more general definition of status. In prior sections, we assume that users strictly prefer having a higher badge than having the same badge as another user. We loosen this assumption about how users value ``ties'' and show that the optimal mechanism can be quite complex in this general setting. However, we show that both of approximately optimal mechanisms that we study maintain their simple structure and performance guarantees in this more complex setting.

%% file: relatedWork.tex


There's a growing literature on the role and design of incentives systems \cite{dubey2010grading,ghosh2013,GM12,dipalantino2009crowdsourcing,easley2013incentives,hopkins2004running,bachrach2013,moldovanu2007contests}.  These papers consider how to award badges, virtual points, a monetary prize, or viewer attention on a website in order to maximize either the total quantity of contributions or the largest contribution. The general sense of the literature suggests that when agents have exogenous values for their rewards, coarse reward systems do well at maximizing these objectives. The optimal choice depends heavily on the modeling assumptions. Most of these papers consider varying informational assumptions about the environment and derive that coarse mechanisms (often extremely coarse mechanisms with just two reward types, winners and losers) are either optimal or approximately optimal.\footnote{Interestingly, \cite{chawla2012optimal} also show that their reward system doesn't ``overproduce'' in the sense that the largest contribution is at least half the size of the total quantity of contributions.}  These informational assumptions include: abilities of users are private~\cite{archak2009optimal,chawla2012optimal} or publicly known~\cite{dubey2010grading}, noisy observations as to the size or quality of the contribution~\cite{ghosh2013, easley2013incentives}.  However, most of these papers have an exogenous well-defined user utility for the associated rewards.\footnote{Two notable exceptions are \cite{moldovanu2007contests,dubey2010grading} who consider a general convex value for the reward.} Our work focuses on relaxing assumptions about the exact functional form of the user utility and instead investigates the impact of the shape of utility on reward systems.  Consistent with prior work, we find that coarse reward systems can be optimal for a wide class of valuations but in contrast to prior work, they can do arbitrarily poorly when valuations are endogenously determined by the scarcity of a reward. 

The paper most closely related to ours is that of \cite{moldovanu2007contests}. We consider the same model of a game where users are motivated by social status. The key difference is that they assume a specific status function. They prove that a fine partitioning of agents is optimal, and show that a coarse partitioning of agents is a 2-approximation (under certain distributional assumptions). We consider general status value functions (theirs is a special case of the class we consider in section \ref{sec:tie-breaking}) and in doing so, develop a more general theory of optimal mechanisms for status contests. Notably, we prove that the optimal mechanism in their setting can be an arbitrarily bad for a large class of status valuations. On a technical level, \cite{hopkins2004running} use the mathematical connection between allocation in a first-price auction and a consumer's expected social status to analyze a game of consumer choice. Although the goal in our paper is different, we extend this connection to a wider class of social status functions and allocations resulting from a wider class of status mechanisms.

%% file: model.tex

We now introduce our game theoretic model of contributions to a user-driven site. Users contribute content to the site but these contributions cost them effort. There is a population of $n$ {\em users} and an ordered set of $m+1$ {\em badges}, where badges are ordered such that $m \succ m-1 ... \succ 1 \succ 0$.  Each user simultaneously makes a contribution $b_i\in\R^+$ to a {\em badge mechanism}.  The badge mechanism maps this profile of contributions to an assignment of badges for each user. Formally, a badge mechanism is a function $\rank:\R^+\times(\R^+)^{n-1} \rightarrow \{1,...,m\}$ where we say user $i$ earns badge $r(b_i,b_{-i})$.  

A user's utility is a function of his {\em status}, which determines his value, and his {\em ability}, which determines his cost.  The {\em status} of a user $i$ is defined as the fraction of users who have earned an equal or better badge. We denote this fraction by $t_i=\frac{|\{j\neq i:~\rank(b_j,b_{-j})\succeq \rank(b_i,b_{-i})\}|}{n-1}$. A user's \emph{status value} is given by a function $\stat{\cdot}:[0,1]\rightarrow \R^+$ of $t_i$. We assume that $S(1) = 0$, i.e.\ users in the lowest status class derive a status value of 0.  The {\em ability} $\val[i]$ of user $i$ is private information and drawn independently and identically from a common distribution F with support over $[0, \bar{\val}]$ and density f.  We assume $F$ is atomless and regular.\footnote{This is a weaker assumption than the {\em monotone hazard rate} condition, assumed in Moldovanu et al. \cite{moldovanu2007contests}; see Section~\ref{sec:preliminaries} for a formal definition.}
 If a user with ability $\val[i]$ contributes $b_i$, then he incurs a cost of $\frac{b_i}{\val[i]}$. A user's utility for contributing $b_i$ is quasi-linear in his status value and his cost of contribution: $$ S(t_i) - \frac{b_i}{\val[i]}.$$

The goal of the designer is to maximize the sum of the contributions of all users $\sum_i b_i$.  We note that we can extend this framework to incorporate objectives such as as maximizing contributions that are at least of a certain quality (to mitigate agents submitting large amounts of low quality content) by redefining what a contribute constitutes. The important assumption about the objective is that the designer wants to maximize the total sum of contributions. \footnote{Alternative objectives could be maximizing the quality of the single best solution such as in \cite{chawla2012optimal} and \cite{archak2009optimal}.}

\textbf{Badge Mechanisms.} 
While our definition of a badge mechanism allows for a number of ways to award social status as a function of contributions, we study a few canonical mechanisms in this paper. 
\begin{definition}[Absolute Threshold Mechanism]\label{def:absolute}
An absolute threshold mechanism is defined by a set of $m$ thresholds $\tbvec=(\tb{1},\ldots,\tb{m})$, with $\tb{1}\leq \ldots \leq \tb{m}$, such that user $i$ is awarded badge $j\in\{0,\ldots,m\}$ if $b_i \in [\tb{j}, \tb{j+1})$. By convention, $\theta_0=0$ and $\theta_{m+1}=\infty$.
\end{definition}
\begin{definition}[Leaderboard Mechanism]\label{def:ranking}
The leaderboard mechanism assigns each user a distinct badge among a set of $m=n$ badges in decreasing order of their contributions. If user $i$ contributes the $j^{th}$ highest amount (where $m$ is highest amount and 1 is the lowest), he is assigned badge $j$. In the event that two users submit equal levels of contributions, the tie is broken randomly. 
\end{definition}
The key difference between the above two mechanisms is the badge that user $i$ receives in an absolute threshold mechanism depends only his own contribution $b_i$ and not on the contributions of the remaining players. By contrast, badge that a player earns in a leaderboard mechanism depends only the position of his contribution within the ordered list of all contributions, but not on the amount of his contribution. The next mechanism we define is an example of a hybrid of these two mechanisms. 
 \begin{definition}[Leaderboard with a Cut-Off]\label{def:leaderBoard-Cutoff}
The leaderboard with a cut-off mechanism is defined by a single threshold $\theta$ such that any user who submits $b_i < \theta$ is assigned badge 0. The remaining users are assigned badges in decreasing order of contributions, as in the leaderboard mechanism. 
\end{definition}
We note that this is not an exhaustive list of mechanisms, nor are we the first to construct such definitions. Another common mechanism is a relative threshold badge, studied in \cite{moldovanu2007contests}, where users are assigned based on their coarse position within the ordered list of contributions. For example, a user will receive the top badge if he is in the top half of contributors and the lowest badge otherwise. We focus on the three mechanisms defined above because they give a good sketch of the properties of (approximately) optimal mechanisms across a range of environments.  


\textbf{Status Value.}
Our goal in this paper is to characterize the optimality of various badge mechanisms for different regimes of status valuations. We divide status valuation functions into the classes of linear functions, concave functions, and convex functions. 
Each regime has a natural interpretation; for concave status, a user's marginal gain in status value decreases as he increases his standing in society. For convex, the marginal gain in status value increases as a user increases his standing, and in linear, the marginal gain in status value is constant.

%% file: preliminaries.tex

In this section, we briefly review a few classic results from optimal mechanism design which we make extensive use of in our paper.

{\bf Optimal Mechanism Design.}  In the standard auction design problem, there are $n$ agents competing for a single unit of a divisible good.  The goal is to design an auction that maximizes {\em revenue} or total payments.  Each agent $i$ has a {\em value} $\val[i]$ per unit of the item drawn IID from an atomless distribution $\cdf$ with support $[0,\bar{\val}]$ and density $\pdf[\cdot]$.  Under this assumption, there is a one-to-one correspondence between an agent's value $\val[i]$ and his {\em quantile} $\quant[i]$:
$$
\quant(\val[i])=1-\cdf(\val[i]) \in [0,1]\mathrm{\ and\ }\val(\quant[i])=F^{-1}(1-\quant[i]).
$$ 
Intuitively, the quantile denotes the probability that a random sample from $\cdf$ has higher ability than agent $i$: thus lower quantile corresponds to higher value.  Note the distribution of quantiles is uniformly distributed in $[0,1]$.  As it is often more convenient to work in quantile space, we will henceforth state results using this terminology.

A direct revelation auction solicits {\em bids} $\bvec=(b_1,\ldots,b_n)$ from the agents and computes an allocation $\{x_i(\bvec)\}$ and set of payments $\{p_i(\bvec)\}$.  The resulting utility for agent $i$ with quantile $\quant[i]$ is then 
\begin{equation}
\label{eqn:standardutil}
u_i(\bvec;q_i)=\val(\quant[i])\cdot x_i(\bvec)-p_i(\bvec).
\end{equation}
Optimal auction design asks what allocation and payment rules maximize the total expected payments of the agents.  Answering this question requires a prediction of agent behavior: i.e., given a quantile and knowledge of the setting (distributions of other agents, auction rules, etc.), how will agents map quantiles to bids?  It is standard to predict that these maps, or {\em bidding strategies}, will comprise a {\em Bayes-Nash equilibrium} (BNE) of the underlying incomplete information game.  A profile of bidding strategies $b_i(\cdot)$ mapping quantiles to bids is a BNE if each agent maximizes their utility in expectation over the quantiles of the rest of the agents: for all $\quant[i]$, $b_i'$:
$$
\E_{\qvec_{-i}}\left[ u_i\left(\bvec(\qvec);\quant[i]\right)\right]\geq \E_{\qvec_{-i}}\left[u_i(b_i',\bvec_{-i}(\qvec_{-i});\quant[i])\right].
$$
This constraint gives rise to the standard characterization of BNE due to Myerson \cite{Myerson81} and Bulow and Roberts \cite{bulow1989simple}.  Letting $\tilde{x}_i(\qvec)=x_i(\bvec(\qvec))$, $\hat{x}_i(\quant[i])=\E_{\qvec}[\tilde{x}_i(\qvec);\quant[i]]$ and $\hat{p}_i(\quant[i])=\E_\qvec[p_i(\bvec(\qvec));\quant[i]]$ denote the {\em ex-post} allocation, {\em interim} allocation and payment rules respectively, the lemma  states:

\begin{lemma}[\cite{Myerson81}, \cite{bulow1989simple}]\label{lem:myerson-bne} A profile of bidding functions $\bvec(\cdot)$ and an implied profile of interim allocation and payment rules ${\bf \hat{x}}(\cdot)$ and ${\bf \hat{p}}(\cdot)$  are a BNE only if:
\begin{itemize}
\item $\hat{x}_i(\quant[i])$ is monotone non-increasing in $\quant[i]$ and
\item $\hat{p}_i(\quant[i]) = \val(\quant[i])\cdot \hat{x}_i(\quant[i]) + \int_{\quant[i]}^{1}\hat{x}_i(z)\cdot \val'(z)\cdot dz$
\end{itemize}
These two conditions are sufficient for $\bvec(\cdot)$ to be a BNE if each bid function $b_i(\cdot)$ spans the whole region of feasible bids. Otherwise,  they only imply that each player doesn't wants to deviate to the set of bids spanned by $b_i(\cdot)$.
\end{lemma}

Simple manipulations of these identities allow one to characterize the revenue of a mechanism concisely in terms of the quantiles of the agents. Let $R(\quant)=\quant\cdot \val(\quant)$ denote the \emph{revenue function}\footnote{Intuitively, the revenue function computes the revenue a seller can generate by selling the (entire) item with probability $q$.} of the value distribution $\cdf[]$, and $R'(\quant[i])$ the \emph{virtual value} of a player.  Then,

\begin{lemma}[\cite{Myerson81}, \cite{bulow1989simple}]\label{lem:myerson-opt}
The expected total payment of a mechanism is equal to the expected {\em virtual surplus}:
$$
\E_{\qvec}\left[\sum_i p_i(\quant[i])\right] = \E_{\qvec}\left[\sum_i R'(\quant[i])\cdot \tilde{x}_i(\qvec)\right].
$$
while the expected payment of each player is his expected virtual surplus allocation: $\E_{\quant[i]}\left[p_i(\quant[i])\right]=\E_{\quant[i]}\left[R'(\quant[i])\cdot \hat{x}_i(\quant[i])\right]$.
\end{lemma}

A consequence of Lemma~\ref{lem:myerson-bne}, known as {\em revenue equivalence}, is that if two mechanisms have the same interim allocation rule in BNE, they will also generate the same revenue.  A consequence of this observation and Lemma~\ref{lem:myerson-opt} is that the optimal mechanism is simply the mechanism that maximizes the expected virtual surplus. This reduces mechanism design to a constrained optimization problem.

A distribution $\cdf[]$ is \emph{regular} if the revenue function $R(\quant)=\quant \cdot \val(\quant)$, is a concave function, or equivalently, the virtual value of a player $R'(\quant)$ is non-increasing in his quantile.  Our assumption that $\cdf[]$ has support $[0,\bar{v}]$ also implies that $R(0)=R(1)=0$. Moreover, since virtual value $R'(\quant)$ is non-increasing in quantile, this implies that the virtual value is positive up until some quantile $\tqstar$, and negative afterwards. Quantile $\tqstar$, where $R'(\tqstar)=0$, is defined as the \emph{monopoly quantile} of the value distribution. Lemma~\ref{lem:myerson-opt} implies that the optimal mechanism only allocates to agents with quantiles below $\tqstar$.

%% file: optimal.tex

We will use the framework of optimal auction design to design the optimal badge mechanism.  To this end, it is useful to work in quantile space and (wlog) scale the utility functions of the users by their abilities, which are just constants from the perspective of each user, when deciding his contribution level.  After doing so, the utility function of a user becomes:
\begin{equation}
u_i(\bvec;\quant[i]) = \val(\quant[i]) \cdot \stat{t_i} - b_i.
\end{equation}
Comparing with Equation \eqref{eqn:standardutil}, we see that the (virtual) ability of a user equates with the (virtual) value of an agent in the auction literature, the status allocation of a user equates with the item allocation and the contribution equates with the payment. The badge mechanism that maximizes contributions is thus functionally equivalent to the optimal (revenue-maximizing) auction.  Note our setting exhibits two key restrictions not present in standard optimal auction design:
\begin{enumerate}
\item The set of feasible allocations (status values) is highly constrained and ill-behaved due to the externalities that a user's status imposes on others.  In particular, the total amount of allocation, i.e., $\sum_i\stat{t_i}$ is not constant as in the standard auction setting discussed in Section~\ref{sec:preliminaries}.
\item The payments (contributions) are not determined by the mechanism, but rather by the users themselves.  In this sense, a badge mechanism is a special type of auction, known in the literature as an all-pay auction.
\end{enumerate}

The optimal auction framework states that the optimal mechanism chooses an ex-post allocation rule which maximizes expected virtual surplus and then computes payments which support the implied interim allocations in equilibrium (see Section~\ref{sec:preliminaries}).  Following this reasoning, we first compute the ex-post virtual surplus-maximizing badge allocation.  We then argue that this allocation is implemented in a BNE by a {\em leaderboard with a contribution cutoff}.  Finally, we show that the derived BNE is unique.

\subsection{Virtual Surplus-Maximizing Badge Allocation}

Maximizing virtual surplus for some instantiation of a quantile profile $\qvec$ is simply an optimization problem, subject to the
constraints that are implicit in the way that users derive status. The optimization problem asks: given a vector of virtual abilities $R'(\quant[1]),\ldots,R'(\quant[n])$, assign badges $\rvec=(\rank_1,\ldots,\rank_n)$ to the users so as to maximize: $\sum_i R'(\quant[i])\cdot \stat{t_i(\rvec)}$. The following lemma states that the solution assigns to users distinct badges in decreasing order of their quantile so long as their quantile is below the {\em monopoly quantile} $\tqstar{}$ (equivalently, so long as their virtual ability is non-negative), and assigns all other users a badge of $0$. The formal proof is deferred to the appendix.

\begin{lemma}[Virtual Surplus Maximizing Badge Allocation]\label{lem:optimal-surplus} 
Let $\quant[1]\leq \ldots \leq \quant[k]\leq \tqstar< \quant[k+1]\leq \ldots \leq \quant[n]$ be a profile of quantiles. Then the optimal virtual surplus is achieved by assigning a distinct decreasing badge to all users $\{1,\ldots,k\}$ with non-negative virtual ability and badge $0$, to all negative virtual ability users $\{k+1,\ldots,n\}$, i.e. $r_1=n, r_2=n-1,\ldots, r_k=n-k+1$, and $r_{k+1}=\ldots=r_n=0$. 
\end{lemma}

\subsection{Implementation}

We now show that the ex-post virtual surplus maximizing allocation of badges is implemented at the unique equilibrium of the \emph{leaderboard with a cutoff} mechanism as defined in section \ref{sec:model}. To do so, we need to show two things: first, we must describe the interim status allocation rule implied by the ex-post allocation rule in Lemma~\ref{lem:optimal-surplus}.  Then we must compute the corresponding equilibrium contributions using the payment identity of the optimal auction framework and check that these contributions do indeed give rise to the interim allocation, under the rules of the aforementioned badge mechanism.

The interim status allocation of a user is the expected status value he receives from the mechanism given his quantile $\quant[i]$. To compute it, let $T_i$ be the random variable denoting the number of opponents with quantile smaller than $\quant[i]$. 
Observe that if $\quant[i]\leq \tqstar$, then under the optimal
ex-post allocation of badges in Lemma \ref{lem:optimal-surplus}, user $i$ will be ranked at the $T_i+1$ position. 
Thus the implied interim status of a user with $\quant[i]\leq \tqstar$ is $\E_{\qvec_{-i}}\left[S\left(\frac{T_i}{n-1}\right)\right]$ and 0 if $\quant[i]>\tqstar$. $T_i$ is distributed as a binomial distribution of $n-1$ independent random trials, each with success probability of $\quant[i]$.\footnote{Recall that quantiles are distributed uniformly in $[0,1]$.} For convenience, we introduce the following notation
\begin{equation}\label{eqn:expectedstatus}
\statn{n}{\quant[i]} = \sum_{\nu=0}^{n-1}\stat{\frac{\nu}{n-1}}\cdot \beta_{\nu,n-1}(\quant[i])
\end{equation}
where $\beta_{\nu,n}(q)= \binom{n-1}{\nu}\cdot q^{\nu}\cdot (1-q)^{n-1-\nu}$, denotes the Bernstein basis polynomial and $\statn{n}{\quant}$ is the Bernstein polynomial approximation of the status function $\stat{\cdot}$. By properties of Bernstein polynomials (see \cite{phillips-bernstein}), if $\stat{\cdot}$ is a strictly decreasing function then $\statn{n}{\cdot}$ is also strictly decreasing and if $\stat{\cdot}$ is convex or concave then so is $\statn{n}{\cdot}$. Additionally, $\statn{n}{\cdot}$ is continuous and differentiable and $\statn{n}{0}=\stat{0}$ and $\statn{n}{1}=\stat{1}=0$. 

Using this notation, we can express the interim status allocation for each user. This expression will be very useful throughout the course of this paper, so we formalize it in the following proposition. 
\begin{proposition}\label{prop:badge-interim}
In the optimal badge mechanism, the interim status allocation of a user with quantile $\quant[i]$ is
\begin{equation}\label{eqn:optimal-interim}
\xq{\quant[i]}=  
\begin{cases}
\statn{n}{\quant[i]} & \quant[i] \leq \tqstar \\ 
0 & \quant[i] > \tqstar
\end{cases}
\end{equation}
\end{proposition}

If the interim allocation of status of a user has the form presented in Equation \ref{eqn:optimal-interim}, then for his contribution to constitute an equilibrium of the badge mechanism, it must satisfy the 
payment characterization of Lemma \ref{lem:myerson-bne}:
\begin{equation}\label{eqn:optimal-bid}
b(\quant) = \begin{cases}
\val(\quant)\cdot \statn{n}{\quant}+\int_{\quant}^{\tqstar}\statn{n}{z}\cdot \val'(z)\cdot dz & \text{if $\quant\leq \tqstar$}\\
0 & \text{if $\quant\geq \tqstar$}
\end{cases}
\end{equation}
Moreover, by Lemma \ref{lem:myerson-opt} the expected user contribution under the optimal mechanism will be:
\begin{equation}\label{eqn:optimal-user-contr}
\opt = \E_{q}[b(q)] = \int_{0}^{\tqstar} R'(\quant)\cdot \statn{n}{\quant}d\quant
\end{equation}
We now show that the above pair of interim allocation and equilibrium contribution, given in Equations \eqref{eqn:optimal-interim} and \eqref{eqn:optimal-bid}, constitute the unique equilibrium of a badge
mechanism that takes the form of a leaderboard with a cutoff.

\begin{theorem}[Optimal Badge Mechanism: Leaderboard with a cutoff]\label{thm:optimal}
The optimal badge mechanism assigns a distinct badge to each user in decreasing order of contribution (breaking ties at random) as long as they pass a contribution threshold of $\theta=\val(\tqstar)\cdot\statn{n}{\tqstar}$. User's that don't pass the contribution threshold are assigned a badge of $0$. The mechanism has a unique equilibrium.
\end{theorem}
\begin{proof}
The proof consists of two parts: First we show that the interim allocation of status $\xq{\cdot}$, given in Equation \eqref{eqn:optimal-interim} and the contribution function given in Equation \eqref{eqn:optimal-bid}, constitute an equilibrium of the proposed badge mechanism. To achieve this we simply need to verify that the implied contribution function, indeed gives rise to the optimal interim allocation of status. Then by Lemma \ref{lem:myerson-bne} and after also checking that users don't want to deviate to contributions outside of the region of contributions spanned by $b(\cdot)$, this pair of $\xq{\cdot}$ and $b(\cdot)$ are an equilibrium. Second, we need to argue that the mechanism has no other equilibria. To achieve this, we show that this specific form of a badge mechanism, falls into the class of anonymous order-based auctions of Chawla and Hartline \cite{Hartline2012}, where it is shown that such auctions have unique equilibria.

Observe that the contribution function in Equation \eqref{eqn:optimal-bid} is strictly decreasing in the region $[0,\tqstar]$, since $b'(\quant)=\val(\quant)\cdot \statnpr{n}{\quant}$, and both $\val(\cdot)$ and $\statn{n}{\cdot}$ are strictly decreasing functions. Moreover, 
observe that $b(\tqstar)=\val(\tqstar)\cdot \statn{n}{\tqstar}=\theta$. Since all quantiles are distributed uniformly in $[0,1]$, under the above bid function, a user with quantile $\quant\leq \tqstar$ is assigned a badge lower than every user with higher quantile and higher than every user with lower quantile, while the event of a tie has zero measure. Thereby, the latter
bid function $b(\cdot)$, gives rise to the optimal interim allocation of status given in Equation \eqref{eqn:optimal-interim}. Last, it is easy to see that since the user with quantile $0$, doesn't want to bid above $b(0)$ (since he gets the same status at a higher cost), no user wants to bid above the bid $b(0)$. The latter follows from the fact that the characterization in Lemma \ref{lem:myerson-bne} guarantees that no user wants to deviate to any contribution in the region of contributions spanned by $b(\cdot)$ and hence $b(0)$. From this it immediately follows that they also don't want to bid above $b(0)$. Additionally, it is also trivial to see that no user wants to bid in the region $(0,\theta)$. Thus we can conclude that the above pair of interim allocation and contribution function are an equilibrium of the mechanism.

The fact that the badge mechanism has a unique equilibrium, follows from the recent results of Chawla and Hartline \cite{Hartline2012}. Despite the fact that general badge mechanisms do not fall into the class of auctions studied in \cite{Hartline2012} (such as for instance absolute threshold badge mechanisms studied later), the specific leaderboard with a cutoff badge mechanism falls into their framework. To show this, we need to argue that for any contribution profile $\bvec$ (even in the case of tied contributions), the ex-post expected allocation of status of a user depends only on his own contribution, on the number of users that have a higher contribution and on the number of users that have the same contribution (but for instance, not on the exact contributions of other users). Since, ties are broken uniformly at random, a user is ranked above another user with an equal contribution with probability of $1/2$. Thus if we denote with $n_g(\bvec)$ the number of users that contribute more and with $n_e(\bvec)$ the number of users with an equal contribution, then the expected ex-post status allocation of a user is:
\begin{equation*}
\sum_{\nu=0}^{n_e(\bvec)}\stat{\frac{n_g(\bvec)+n_e(\bvec)}{n-1}}\cdot \beta_{\nu,n_e(\bvec)}(1/2)
\end{equation*}
when his contribution $b_i\geq \theta$ and $0$ otherwise. It is clear that this ex-post allocation only depends on the quantities required by the framework of \cite{Hartline2012}.
\end{proof}

\paragraph{Graphical interpretation of optimal average and ex-post per-user contribution} To provide more intuition of what is the average contribution of each user under the optimal badge mechanism, we analyze the case of a linear status function $\stat{t}=1-t$. In this case, by linearity of expectation, we also get that $\statn{n}{\quant}=1-\quant$. Therefore, the 
optimal interim status allocation of a user is:
\begin{equation}\label{eqn:linear-optimal-interim}
\xq{\quant}=  
\begin{cases}
1-\quant & \quant \leq \tqstar \\ 
0 & \quant > \tqstar
\end{cases}
\end{equation}
while his equilibrium contribution after applying integration by parts in Equation \eqref{eqn:optimal-bid} is:
\begin{equation}\label{eqn:linear-optimal-bid}
b(\quant) = \begin{cases}
\val(\tqstar)(1-\tqstar)+\int_{\quant}^{\tqstar}\val(z)\cdot dz & \text{if $\quant\leq \tqstar$}\\
0 & \text{if $\quant\geq \tqstar$}
\end{cases}
\end{equation}
Last, by Equation \ref{eqn:optimal-user-contr} and applying integration by parts, the
expected contribution of each user is:
\begin{equation}\label{eqn:linear-contribution}
\opt=\E_{\quant}\left[b(\quant)\right] = \int_0^{\tqstar}R'(\quant)\cdot(1-\quant)\cdot d\quant = R(\tqstar)\cdot(1-\tqstar)+\int_{0}^{\tqstar}R(\quant)d\quant
\end{equation}
In this case, the expected user contribution has a nice pictorial representation, if we plot the revenue function of the ability distribution, as is depicted in Figure \ref{fig:optimal}.
\begin{figure}[htpb]
\centering
\subfigure[Optimal Mechanism]{\includegraphics[scale=.6]{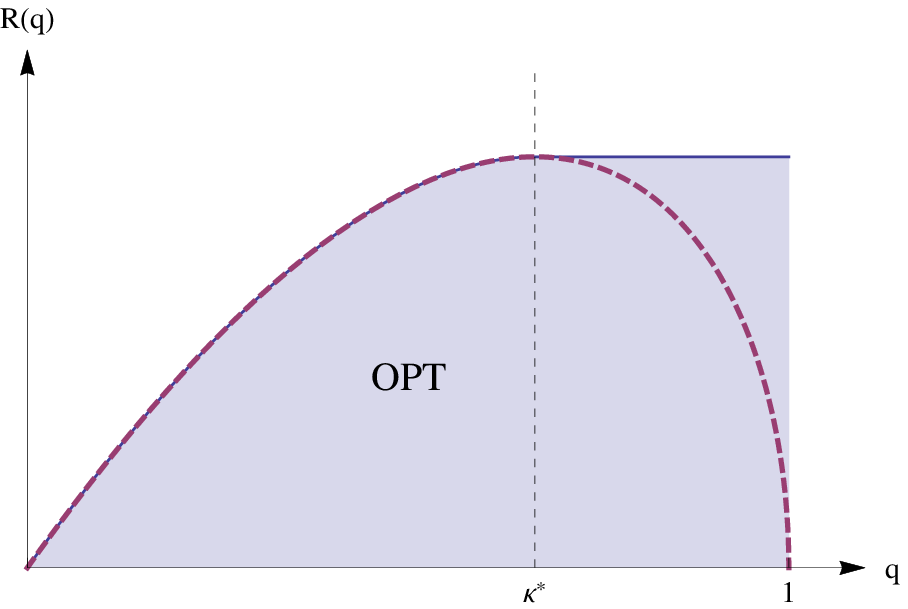} } 
\subfigure[Absolute Threshold Mechanism]{\includegraphics[scale=.6]{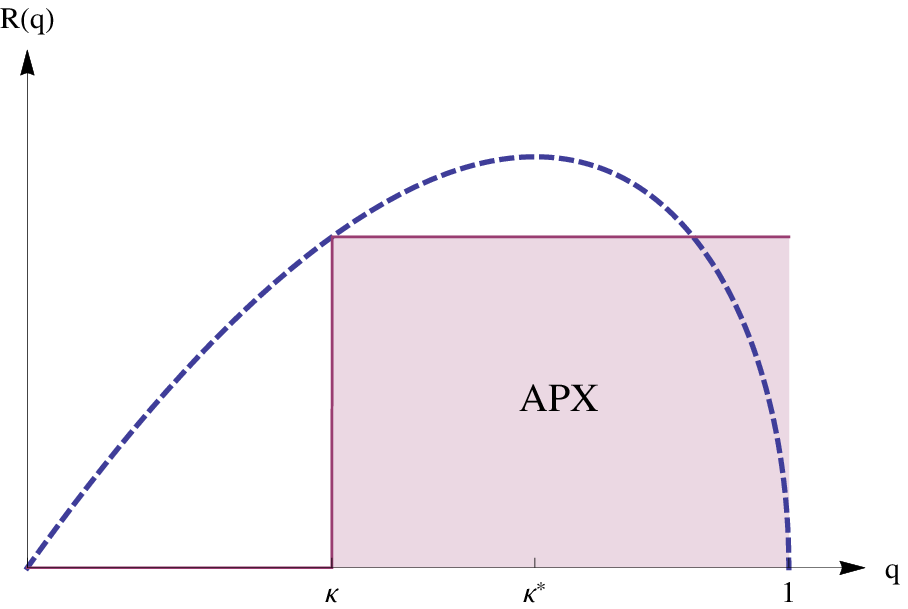}}
\caption{  
These figures plot the revenue curves as a function of quantile when the distribution of abilities is given by the CDF $F(v) = v^2$ (equivalently $\val(\quant)=\sqrt{1-\quant}$). When status is linear, the expected user contribution of the optimal mechanism is given by the shaded area in the left figure. The expected user contribution of an absolute threshold mechanism with a single threshold, as discussed in section \ref{sec:single} is given by the shaded area in the figure on the right. \label{fig:optimal}}
\end{figure}

\paragraph{Drawbacks of optimal badge mechanism} The optimal badge mechanism described above has two main drawbacks. First the equilibrium contribution is a very complex function of the ability distribution, and therefore, small errors in a user's beliefs will dramatically change the total contribution that would arise in practice. Second, the optimal mechanism is not detail-free, as it depends on the monopoly quantile of the distribution. In the sections that follow we will address each of these issues separately and we will give approximately optimal badge mechanisms that avoid some of these drawbacks. 
%

%% file: single.tex
In this section, we explore the approximate optimality of absolute threshold mechanisms as formally described in Definition \ref{def:absolute}. We show that for concave status valuations (including linear), the expected user contribution of a mechanism with a single threshold is a 4-approximation to the expected user contribution generated by the optimal mechanism. This threshold is set such that any user with ability greater than the median ability will earn the top badge while the other half of the population will not contribute anything. Furthermore, an absolute threshold mechanism with a slightly more complex contribution threshold, which intuitively takes into account the dispersion of the ability distribution, is a 2-approximation for linear status and a 3-approximation for strictly concave status. By way of example, we show that no mechanism with a single threshold can do better than a 2-approximation. When status valuations are convex, no absolute threshold mechanism with a constant number of thresholds can achieve any finite approximation and thus the designer must add more thresholds as the number of users grows in order to achieve a good approximation. However, on the positive side, we show that the number of thresholds grows quite slowly with the number of users, indicating that absolute mechanisms with a ``small'' number of thresholds can be approximately optimal. 

We begin our analysis by characterizing the unique symmetric Bayes-Nash equilibrium\footnote{We note that asymmetric equilibria do exist in absolute threshold badge mechanisms.} of an absolute threshold mechanism with thresholds $\tbvec$. This equilibrium is monotone, non-increasing in each user's quantile. A lower quantile implies a higher ability, so higher ability users will contribute more than lower ability users in equilibrium. 

\begin{theorem}[Equilibrium Structure]\label{thm:uniqueness}
Any absolute threshold mechanism with contribution thresholds $\tbvec=(\tb{1},\ldots,\tb{m})$ has a unique symmetric BNE $b(\cdot)$ characterized by a vector of quantile thresholds $\tqvec=(\tq{1},\ldots,\tq{m})$ with $\tq{1} \geq \tq{2} ... \geq \tq{m}$ such that a user with ability quantile $q_i$ will make a contribution of: 
\begin{eqnarray*}
b(\quant[i]) = 
\begin{cases}
0 & \quant[i] > \tq{1} \\
\theta_t & \quant[i] \in (\tq{t+1},\tq{t}] \\
\theta_m & \quant[i] \leq \tq{m} 
\end{cases}
\end{eqnarray*}
These quantile thresholds can be computed by a system of $m$ equations.


Furthermore, for any vector of quantile thresholds $\tqvec=(\tq{1},\ldots, \tq{m})$ there exists a vector of contribution thresholds $\tbvec$, characterized by:
\begin{align}\label{eqn:effort-to-value}
\forall t\in[1,\ldots, m]: \tb{t} =\sum_{j=1}^t \val(\tq{j})\cdot \left(\statn{n}{\tq{j}}-\statn{n}{\tq{j-1}}\right),
\end{align}
under which the unique symmetric BNE implements $\tqvec$. 
\end{theorem}

In the unique symmetric BNE, a vector of contribution thresholds gives rise to a vector of quantile thresholds. The main implication of this theorem is that we can design mechanisms over quantile thresholds and know that there exists a vector of contribution thresholds that implement the desired quantile thresholds. This greatly simplifies the problem of optimal design because we no longer need to worry about the induced equilibrium behavior. For the remainder of this section, we focus on designing badge mechanisms in quantile space rather than contribution space.

\subsection{Expected Contributions of a Single Threshold Mechanism}

We now analyze the approximate optimality of absolute threshold mechanisms that use a single threshold. We start by deriving the expected user contribution generated by any such badge mechanism and then prove that, if the threshold is set appropriately, then it is a good approximation to the expected contribution under the optimal mechanism, for any concave status function.

Let $\theta$ be the contribution threshold and let $\tq{}$ be the corresponding equilibrium quantile threshold. There are only two rational contributions in equilibrium, $b_i = \theta$ or $b_i = 0$. At equilibrium all users with quantile smaller than $\tq{}$ will contribute $\theta$ and get the top badge. Hence, the interim status allocation of a player who gets the top badge is $\statn{n}{\tq{}}$. Any user with quantile $\quant[i] = \tq{}$ must be indifferent between contributing $\theta$ and earning the top badge, or contributing 0 and earning badge 0. Thus $\tq{}$ and $\theta$ must satisfy the following relationship: 
\[ \theta = v(\tq{})\cdot \statn{n}{\tq{}} \]
Then the expected contribution of a user is equal to the probability that she has a quantile $\quant[i] \leq \tq{}$, times $\theta$. 
\begin{equation}\label{singleRev}
\apx = E[b_i] = \tq{}\cdot \theta =  \tq{} \cdot v(\tq{}) \cdot \statn{n}{\tq{}} = R(\tq{})\cdot \statn{n}{\tq{}}
\end{equation}

This gives a simple expression for the expected user contribution  of a single threshold mechanism as a function of the quantile threshold implemented in equilibrium. In the case of linear status, the expected user contribution is graphically represented in right plot in Figure \ref{fig:optimal}. This representation serves as an intuition for the upper and lower bounds on approximations achievable by single absolute threshold mechanisms.


\begin{theorem}[Median Absolute Badge Mechanism]\label{cor:median-linear}
When the status function $\stat{\cdot}$ is concave, the expected total contribution of an absolute threshold mechanism with a single quantile threshold $\tq{}=\frac{1}{2}$ is a $4$-approximation to the expected total contribution of the optimal badge mechanism.
\end{theorem}

This suggests that designers can implement a good incentive mechanism by setting the threshold of a single badge such that half of the user base earns the badge.The next theorem shows that this approximation can be improved by incorporating the monopoly quantile of the ability distribution $\tqstar$. In the Appendix we show that the bound of $4$ for the median badge mechanism is tight. Hence, incorporating the monopoly quantile is essential in getting better approximation guarantees.

\begin{theorem}\label{thm:better-apx}
When the status value function $\stat{\cdot}$ is concave, the expected total contribution of an absolute threshold mechanism with a single quantile threshold $\tq{}=\min\{\tqstar,1/2\}$ is a 3-approximation to the expected total contribution of the optimal badge mechanism. Furthermore, when the status value function is linear, it is a 2-approximation.
\end{theorem}

%% file: lowerBound.tex

%

We now explore the limitations of absolute threshold mechanisms with a single threshold. Example \ref{ex:singleConcave} shows that no absolute threshold mechanism with a single threshold can yield better than a 2-approximation to the optimal mechanism even when the status valuation is linear. Example \ref{convexLowerBound} shows that for the class of convex status functions, no absolute threshold mechanism with even a constant number of thresholds can achieve a constant approximation to the optimal mechanism.

\begin{example}\label{ex:singleConcave}[Tight lower bound for any single absolute threshold] 
\begin{figure}[htpb]
\centering
\subfigure{\includegraphics[scale=.6]{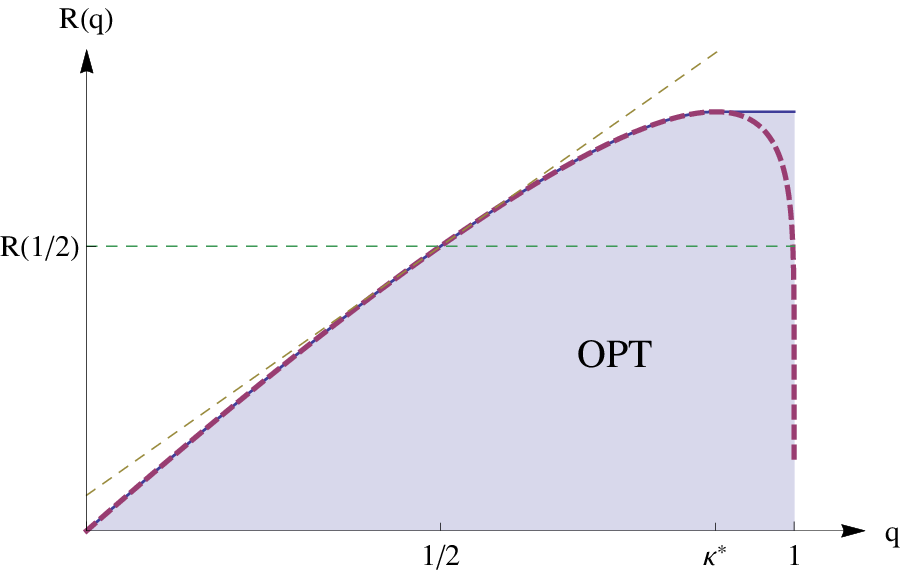}}
\subfigure{\includegraphics[scale=.6]{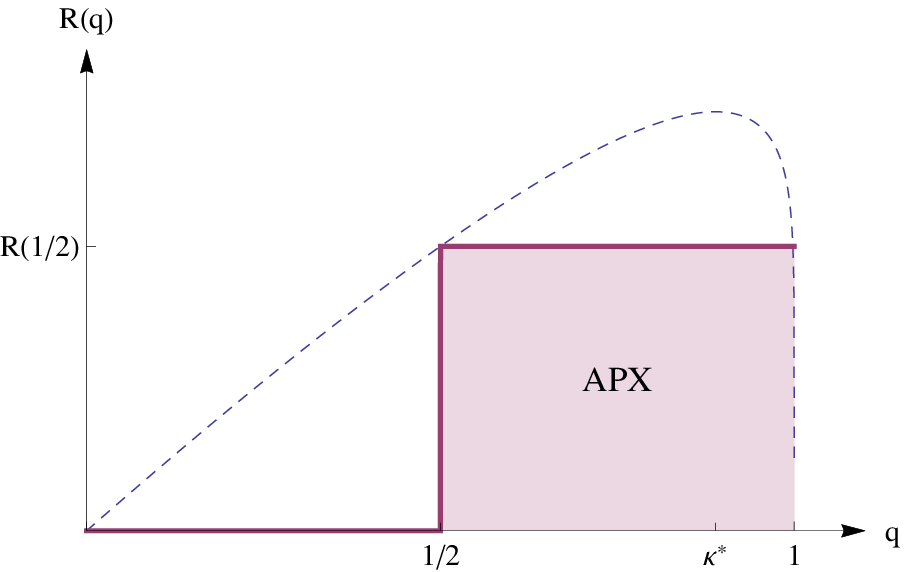}}
\caption{Left figure depicts the per user contribution of the optimal badge mechanism, when the ability distribution is of the form $\cdf[\val]=\val^{\alpha}$. Observe that the rectangle
with height $R(1/2)$ is at least as large as the area below the concave curve, since the curve lies below the
tangent at $1/2$. The per user contribution of a single absolute badge at $\kappa=1/2$ is depicted on the right and is half of the latter rectangle. As $\alpha\rightarrow \infty$, the left area converges to a triangle, and no single badge can achieve more than a half approximation.} \label{fig:linear-2}
\end{figure}

Let the status value function be $\stat{t}=1-t$. Figure \ref{fig:linear-2} shows an example of a revenue curve for some distribution with cumulative density $F$. The figure on the left shows the expected user contribution for the optimal mechanism while the right figure shows the expected user contribution for an absolute threshold mechanism with a single threshold. We will construct a distribution $F$ such that the area under the curve on the left converges to a ``triangle'' while the expected contribution of the single threshold mechanism will be the ``rectangle'' on the right. Any rectangle inscribed inside of this triangle will have at most half of the area of the triangle.   

Suppose that the ability distribution has cumulative density function $F(v)=v^{\alpha}$ and support $[0,1]$, as $\alpha\rightarrow \infty$.  The ability as a function of the quantile is then $\val(\quant)=F^{-1}(1-\quant)=(1-\quant)^{1/\alpha}$. Recall that the revenue function is $R(\quant) = \quant \cdot \val(\quant)$, so in this case 
\[ \lim_{\alpha \rightarrow \infty} R(q) = q \] 

All players have positive virtual ability (since $R'(\quant) \geq 0$ for all $\quant$), so the expected user contribution of the optimal mechanism converges to $\opt=\int_0^1 \quant\cdot d\quant = \frac{1}{2}$. On the other hand, the expected user contribution of any single absolute badge mechanism is $\apx=\tq{} \cdot (1-\tq{})\leq \frac{1}{4}$.


Intuitively, the single threshold mechanism has the following limitation. If the quantile threshold is set low, then only high ability users will earn the top badge and the mechanism loses many contributions because a large fraction of users contribute nothing. If the quantile threshold is set high, then a large fraction of users will earn the top badge but the status value of the badge decreases and thus the amount that users are willing to contribute to earn it decreases. The optimal mechanism does not have this drawback. At a high level, a single badge threshold, unlike a complete ranking, is not effective in motivating a population of users with almost identical abilities.
%
\end{example}

We now turn to the case of convex status functions and show that the concavity assumption a necessity for single absolute threshold mechanisms. Specifically, we  
give an example where no absolute threshold threshold mechanism with a constant number of thresholds can achieve any constant approximation to the optimal badge mechanism.
Instead, the number of thresholds needed to achieve a constant approximation
grows logarithmically with number of participants.

\begin{example}[Logarithmic loss for convex status]\label{convexLowerBound}
Consider a convex status function where the status of a user is inversely proportional to the proportion of users (including the user himself) with a weakly better status class (normalized so that $\stat{1}=0$):
\begin{equation}
\stat{t} =\frac{n}{(n-1)t+1}-1
\end{equation}

Assume that abilities are distributed uniformly in $[0,1]$. The ability function is then $\val(\quant)=1-\quant$ and the revenue function is $R(\quant)=\quant(1-\quant)$, with derivative $R'(\quant)=1-2\cdot \quant$ and the monopoly quantile is $1/2$. The interim status allocation of a user with quantile $\quant$ under the optimal mechanism ends up being after simplifications $\statn{n}{\quant}=\frac{1-(1-\quant)^n}{\quant}-1$. Thereby the optimal expected user contribution is:
\begin{align*}
\opt =~& \int_{0}^{1/2}(1-2\cdot \quant)\statn{n}{\quant}d\quant = \int_{0}^{1/2}(1-2\cdot \quant)\frac{1-(1-\quant)^n}{\quant}d\quant-\frac{1}{2}\int_0^{1/2}(1-2\cdot \quant)d\quant \\
\geq~& \int_{0}^{1/2}(1-2\quant)\frac{1-e^{-n \quant}}{\quant}d\quant-\frac{1}{8}= \Theta(\log(n))
\end{align*}

On the other hand we note that the virtual surplus, and hence total contribution, achievable by any mechanism that uses $m$ badges is at most $n\cdot (m-1)$. Even if all users have a maximum virtual ability of $1$, the virtual surplus from any badge mechanism with $m$ badges at any contribution profile $\bvec$ is simply:
\begin{equation*}
\sum_{t=1}^{m} |i: \rank(b_i,\bvec_{-i})=t|\left(\frac{n}{ |i: \rank(b_i,\bvec_{-i})\geq t|}-1\right)\leq \sum_{t=1}^{m} n- \sum_{t=1}^{m}  |i: \rank(b_i,\bvec_{-i})=t|= n\cdot m-n
\end{equation*}
Thus as $n\rightarrow \infty$ the approximation to the optimal total contribution achievable with $m$ badges grows as $O\left(\frac{\log(n)}{m-1}\right)$.
\end{example}


%

%% file: many.tex
The previous examples show the limitations of absolute threshold mechanisms with a single threshold. In section \ref{sec:many-badges} in the appendix, we show how these lower bound examples can be circumvented by using more than one threshold for concave status functions and by using more than a constant, but still small, number of badges for convex status. 

%% file: complete-ranking.tex

In this section we explore the approximation power of the leaderboard mechanism. Recall that the leaderboard mechanism, as defined in section \ref{sec:model}, assigns each user a unique badge based only on the rank of his contribution $b_i$ within the entire vector of contribution $\bvec$. We prove that the leaderboard mechanism is a good approximation to the optimal mechanism when the status function is convex. In contrast, for concave status functions the leaderboard mechanism may be an unboundedly bad approximation.\footnote{An alternative way of bypassing this high inefficiency, is instead of ranking all players, simply rank a top percentile. By setting the top percentile to approximate the monopoly quantile, then intuitively as the number of players grows large, the interim allocation of this mechanism, converges to that of a leaderboard with a cutoff. In this section we focus on the completely prior-free mechanism of ranking everyone.}

\begin{theorem}\label{thm:leaderboard-apx}
For any convex status function, the leaderboard mechanism is a $2$-approximation to the optimal mechanism.
\end{theorem}

On the other hand the following example shows that for concave status functions, the total contribution achieved by the leaderboard mechanism, can be arbitrarily worse than the optimal total contribution. 
This demonstrates the necessity of having some sort of cut-off, below which all agents are assigned the zero badge. 


\begin{example}
Consider the concave status function of $S(t)=(1-t)^{\alpha}$ with a uniform $[0,1]$ distribution of abilities, i.e., $F(x)=x$.  We consider $\alpha\rightarrow 0$ such that the status function is an almost constant function.  Intuitively, this means that a player is very easily satisfied by being simply ranked above a small portion of the population and any other portion yields almost no extra status value. We also consider the number of players $n\rightarrow\infty$ implying that $\statn{n}{\quant}\rightarrow (1-\quant)^{\alpha}$. In this setting, the revenue function $R(\quant)=\quant\cdot F^{-1}(1-\quant)=\quant(1-\quant)$ and the expected per-player contribution of the leaderboard mechanism converges to $0$:
\begin{align*}
\lim_{\alpha\rightarrow0}\lim_{n\rightarrow\infty}\apx =~& \lim_{\alpha\rightarrow0}\int_{0}^{1} R'(\quant)\cdot \lim_{n\rightarrow\infty}\statn{n}{\quant} d\quant
=~ \lim_{\alpha\rightarrow0}\int_0^{1} (1-2\cdot \quant)\cdot (1-q)^{\alpha}d\quant\\
=~& \lim_{\alpha\rightarrow0}\frac{\alpha }{\alpha ^2+3 \alpha +2}\rightarrow 0
\end{align*}
where the interchange of the limit and the integration is justified by the uniform convergence of Bernstein polynomials.  On the other hand the optimal mechanism is a leaderboard with contribution threshold of $\tqstar=1/2$ yielding an expected per-player contribution which converges to a constant of $1/4$:
\begin{align*}
\lim_{\alpha\rightarrow0}\lim_{n\rightarrow\infty}\opt =~& \lim_{\alpha\rightarrow0}\int_{0}^{1/2} R'(\quant)\cdot \lim_{n\rightarrow\infty}\statn{n}{\quant}d\quant
=~\lim_{\alpha\rightarrow0}\int_0^{1/2} (1-2\cdot \quant)\cdot (1-q)^{\alpha}d\quant\\
=~& \lim_{\alpha\rightarrow0}\frac{\alpha +2^{-\alpha -1}}{\alpha ^2+3 \alpha +2}\rightarrow \frac{1}{4}.
\end{align*}
\end{example}

%% file: tie-breaking.tex

Implicit in our utility model so far is that users treat people in the same class as them, as if they were losing to them, since
equally ranked opponents affect a user's utility in an equal manner as opponents ranked strictly higher. How do our results change qualitatively if instead users treated ties differently? We show that for the case of linear status, our main approximation results carry over irrespective of the way that people treat ties. Specifically, we show that the median badge mechanism is a $4$-approximation, for any tie-breaking rule and in fact is implemented at equilibrium by the same contribution threshold. Moreover, the prior-free leaderboard mechanism, is always a $2$-approximation to the optimal mechanism. We view this as an extra robustness property of our simple vs optimal results.

More formally, let $t_e$ be the proportion of opponents that have the same status class as user $i$ and $t_g$ the proportion that have strictly higher status class. 
Then the status value of a user is simply:
\begin{equation}
S(t_e,t_g) = 1-t_e - \beta\cdot t_g
\end{equation}
where $\beta\in [0,1]$. Intuitively, $\beta$ can be seen as the probability of losing against an equally high opponent. Our initial model corresponds to the case 
of $\beta=1$.  Other than that, the utility of a user is the same as defined in Section \ref{sec:model}, i.e. if $t_e^i(b), t_g^i(b)$ is the corresponding proportions for user $i$ under a bid profile $b$, then:
\begin{equation}
u_i(b;\quant[i]) = S(t_e^i(b),t_g^i(b))-\frac{b_i}{\val(\quant[i])}
\end{equation}

\emph{Interpretation of tie-breaking probability in Q\&A forums.} The linear status model with arbitrary tie-breaking probability is rather meaningful in the context of Q\&A forum where status can be interpreted as argumentative power. We can think of status utility as the probability that a user wins in an argument against a random other user from the system. In most such systems, the credibility derived from having a badge plays a quintessential role in deciding who wins the argument (e.g. voted best answer, presented at the top of the page). Our status utility assumes that a user with a strictly higher badge status will win an argument. Additionally, the tie-breaking rule in the status utility function corresponds to the probability that a user wins against an opponent with equal status. Our model so far assumed that in such battles, even if
there is a winner, his victory is Pyrrhic, and no user ends up getting any utility. The alternative model corresponds to 
each user winning with some probability or equivalently deriving lower utility when winning.

\emph{Single absolute threshold approximation.} We first show that the \emph{median badge mechanism} (i.e. an absolute threshold mechanism such that at equilibrium half of the population gets the high badge), achieves a $4$ approximation to the contribution of the optimal mechanism irrespective of the value of $\beta$ and the median mechanism is implemented by setting a contribution threshold of $\frac{v(1/2)}{2}$, again irrespective of $\beta$. We point out that the optimal mechanism changes as $\beta$ varies. We explore the structure of the optimal mechanism in the Appendix (Section \ref{sec:appendix-tie-breaking}), where we point that for any value of $\beta$ other than $0$, $1/2$, and $1$, the optimal mechanism has a very complex structure and doesn't correspond for instance to some ranking mechanism.
\begin{theorem}\label{thm:median-tie-breaking} For any $\beta\in [0,1]$, the \emph{median badge mechanism} achieves a $4$ approximation to the contribution of the optimal mechanism. The \emph{median badge mechanism} is implemented at the unique symmetric monotone equilibrium of the absolute threshold mechanism defined by a contribution threshold of $\frac{v(1/2)}{2}$, when $\beta\geq 1/2$ and at some equilibrium for $\beta<1/2$.
\end{theorem}

\emph{Approximation with leaderboards.} Despite the complex structure of the optimal mechanism for arbitrary $\beta$, we show that for any such $\beta$, the mechanism that complete ranks all users in decreasing order of contribution
(breaking ties uniformly at random), always achieves a $2$-approximation to the optimal direct mechanism at the unique equilibrium.
\begin{theorem}\label{thm:complete-tie-breaking}
For any $\beta\in [0,1]$, the badge mechanism that assigns a distinct badge to each user in decreasing order of contribution, is always a $2$-approximation to the total contribution of the optimal mechanism.
\end{theorem}

\emph{Structure of optimal mechanism.} Observe that the structure of the optimal mechanism changes as the tie-breaking rule varies. Specifically, Lemma \ref{lem:optimal-surplus}
that characterizes the virtual surplus maximizing allocation of badges is no longer valid if $\beta\neq 1$. By the equivalence of revenue maximization and virtual surplus maximization
discussed in Section \ref{sec:optimal}, to characterize the optimal mechanism it suffices to characterize the virtual surplus maximizing allocation. In this Appendix, we show that the virtual surplus maximizing allocation implies a nice structure
of the optimal mechanism only for the case of $\beta=0$, $1/2$ or  $1$, and for other values of $\beta$, the ex-post virtual
surplus maximizing allocation is a complex function of the specific instantiation of player quantiles. This complexity of the optimal badge mechanism, render the extension of our approximate results in the section, even more compelling.

%% file: conclusion.tex

In this paper, we studied the design of badge mechanisms when the valuation for a badge is driven by its ability to impart social status to its owners. The shape of the status value functions dictates the necessary features of any good mechanism. When the status valuations are concave, it is necessary to use coarse partitions and group all the low ability users into the same badge. Under convexity, is is necessary to use fine portioning and separate the high ability users into their own badges. 

This paper is a first step towards combining the game-theoretic study of virtual incentive mechanisms to more realistic utility models. The next step in this agenda is to incorporate a ``warm glow'' component, that is the implicit value users get from contributing content, into the utility function, and to understand how this modification influences the design of optimal incentive mechanisms.  

%% file: appendix-opt.tex
\begin*{}{LEMMA \ref{lem:optimal-surplus}}
\ \ \em Let $\quant[1]\leq \ldots \leq \quant[k]\leq \tqstar< \quant[k+1]\leq \ldots \leq \quant[n]$ be a profile of quantiles. Then the optimal virtual surplus is achieved by assigning a distinct decreasing badge to all users $\{1,\ldots,k\}$ with non-negative virtual ability and badge $0$, to all negative virtual ability users $\{k+1,\ldots,n\}$, i.e. $r_1=n, r_2=n-1,\ldots, r_k=n-k+1$, and $r_{k+1}=\ldots=r_n=0$. 
\end*{}
\\ \\ 
\begin{proofof}{Lemma \ref{lem:optimal-surplus}}
The statement follows by the following arguments: first it is easy to see that the allocation of badges should be monotone non-decreasing in the virtual ability or equivalently non-increasing in quantile, since if $R'(\quant[i])>R'(\quant[j])$ (i.e. $\quant[i]<\quant[j]$) and $r_i<r_j$ then we can increase virtual welfare by swapping the status class of player $i$ and player $j$ (this wouldn't affect the status allocation of the remaining players). Therefore, it holds that $\rank_1\geq \rank_2\geq\ldots,\geq \rank_n$ and it remains to show that $\rank_1>\ldots>\rank_k>0$ and that $\rank_{k+1}=\ldots=\rank_n=0$.

If for some $i<k$ it holds that $\rank_i=\rank_{i+1}=\ldots=\rank_{i+t}$, then by discriminating player $i$ above the remaining players, will increase virtual welfare. More concretely, by setting $\rank_j' = \rank_j-1$ for all $j>i$  (for a moment let's allow for negative badges, since at the end we can always shift the badge numbers), then the satisfaction of all players $j>i$ doesn't change since the number of people that have badge at least as high as them remains the same. Additionally, the status allocation of player $i$ strictly increases, since the number of people ranked at least as high as him, strictly decreased. A recursive application of this reasoning implies that $\rank_1>\ldots>\rank_k>\rank_{k+1}\geq \ldots \geq \rank_n$.

Now we show that it must be that $\rank_{k+1}=\ldots =\rank_n$. By such a grouping, the status allocation of all negative virtual ability players is $0$, and therefore their negative virtual ability is not accounted in the virtual welfare. Moreover, by grouping together negative virtual ability players, the status allocation of all non-negative virtual ability players is unaffected.
\end{proofof}

%% file: appendix-equilibrium.tex

\begin*{}{THEOREM \ref{thm:uniqueness} }
\ \ \em Any absolute threshold mechanism with contribution thresholds $\tbvec=(\tb{1},\ldots,\tb{m})$ has a unique symmetric BNE $b(\cdot)$ characterized by a vector of quantile thresholds $\tqvec=(\tq{1},\ldots,\tq{m})$ with $\tq{1} \geq \tq{2} ... \geq \tq{m}$ such that a user with ability quantile $q_i$ will make a contribution of: 
\begin{eqnarray*}
b(\quant[i]) = 
\begin{cases}
0 & \quant[i] > \tq{1} \\
\theta_t & \quant[i] \in (\tq{t+1},\tq{t}] \\
\theta_m & \quant[i] \leq \tq{m} 
\end{cases}
\end{eqnarray*}
These quantile thresholds can be computed by a system of $m$ equations.


Furthermore, for any vector of quantile thresholds $\tqvec=(\tq{1},\ldots, \tq{m})$ there exists a vector of contribution thresholds $\tbvec$, characterized by:
\begin{align}\label{eqn:effort-to-value}
\forall t\in[1,\ldots, m]: \tb{t} =\sum_{j=1}^t \val(\tq{j})\cdot \left(\statn{n}{\tq{j}}-\statn{n}{\tq{j-1}}\right),
\end{align}
under which the unique symmetric BNE implements $\tqvec$. 
\end*{}
\\ \\
\begin{proofof}{Theorem \ref{thm:uniqueness}}
We start by observing that if a player gets status class $r_i$, her output should be exactly the threshold to win that badge, $\tb{r_i}$,  because producing more output would cost more effort and would not increase her value. Additionally, it is easy to see that the equilibrium mapping would be monotone in quantile, i.e. if $\quant[1]>\quant[2]$ then $b(\quant[1])\leq b(\quant[2])$. In other words if a player with quantile $\quant[1]$ bids $\tb{r_1}$ and with value $\quant[2]<\quant[1]$ he bids $\tb{r_2}$ then it must be that $r_1\leq r_2$. \footnote{Suppose the contrary. For simplicity, denote with $x_{r}$ the expected status that a player gets from bidding $\tb{r}$, assuming the rest of the players follow strategy $b(\cdot)$. Since, $r_1>r_2$ we have $\tb{r_1}> \tb{r_2}$. Since $b(\quant[1])$ is an equilibrium for a player with value $\quant[1]$, it must be that $\val(\quant[1]) (x_{r_1}-x_{r_2})\geq \tb{r_1}-\tb{r_2}>0$. Thus $x_{r_1}>x_{r_2}$ and since $\val(\quant[2])>\val(\quant[1])$ we get: $\val(\quant[2])(x_{r_1}-x_{r_2})>\val(\quant[1])(x_{r_1}-x_{r_2})\geq  \tb{r_1}-\tb{r_2}$. But the latter implies, $\val(\quant[2])\cdot  x_{r_1} -\tb{r_1}>\val(\quant[2])\cdot x_{r_2}-\tb{r_2}$ and 
therefore $b(\val[2])$ cannot be an equilibrium for a player with quantile $\quant[2]$.}

Since the equilibrium mapping is a monotone step function of quantile, it is defined by a set of thresholds in the quantile space $\tq{1},\ldots,\tq{p}$, for some $p\leq m$, such that if player $i$ has quantile $\quant[i]\in (\tq{t+1},\tq{t}]$ then he produces output $b(\quant[i])=\tb{t}$. If $\quant[i] >\tq{1}$ then $b(\quant[i])=0$ and if $\quant[i] \leq \tq{p}$ then
$b(\quant[i])=\tb{p}$. For notational convenience we will denote with $\tq{0}=1$ and $\tq{p+1}=0$. Observe that it is not necessarily true that $p=m$, since some contribution thresholds might be too high.

To characterize the BNE, it remains to compute those quantile thresholds and show that they are unique. A player's status value is a function of the proportion of other players with a weakly better badge. A player with quantile $\quant[i]\in [\tq{t+1},\tq{t}]$ earns the badge associated with quantile $\tq{t}$; thus, because the equilibrium is a monotone step function, any player that has a quantile less than $\tq{t}$ will earn a weakly better badge than player i. By definition, a player has a lower quantile than $\tq{t}$ with probability $\tq{t}$. This allows us to compute the interim status value of player $i$ with quantile $\quant[i]\in [\tq{t+1},\tq{t}]$ as: 
\begin{equation}
\xq{\quant[i]} =  \statn{n}{\tq{t}}
\end{equation}
where $\statn{n}{\cdot}$ is given by Equation \eqref{eqn:expectedstatus}. 

By the payment identity of Lemma \ref{lem:myerson-bne}, for the vector of quantiles $\tqvec$ to be an equilibrium they must satisfy the following equation:
\begin{align}\label{eqn:effort-to-value}
\forall t\in[1,\ldots, p]: \tb{t} =\sum_{j=1}^t \val(\tq{j})\cdot \left(\statn{n}{\tq{j}}-\statn{n}{\tq{j-1}}\right)
\end{align}
This relationship is depicted in Figure \ref{fig:interim-characterization}. 
\begin{figure}[htpb]
\centering
\subfigure{\includegraphics[scale=0.7]{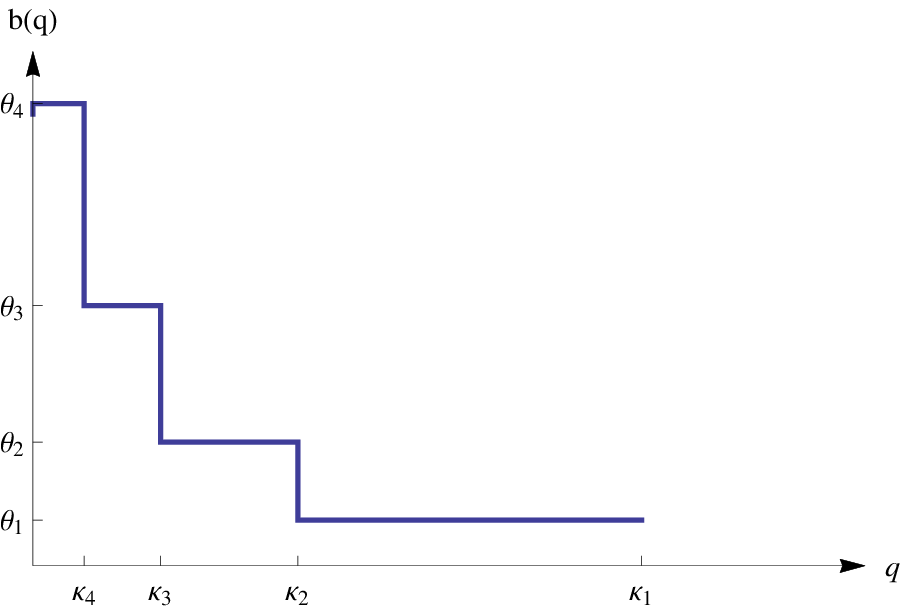}}
\subfigure{\includegraphics[scale=0.7]{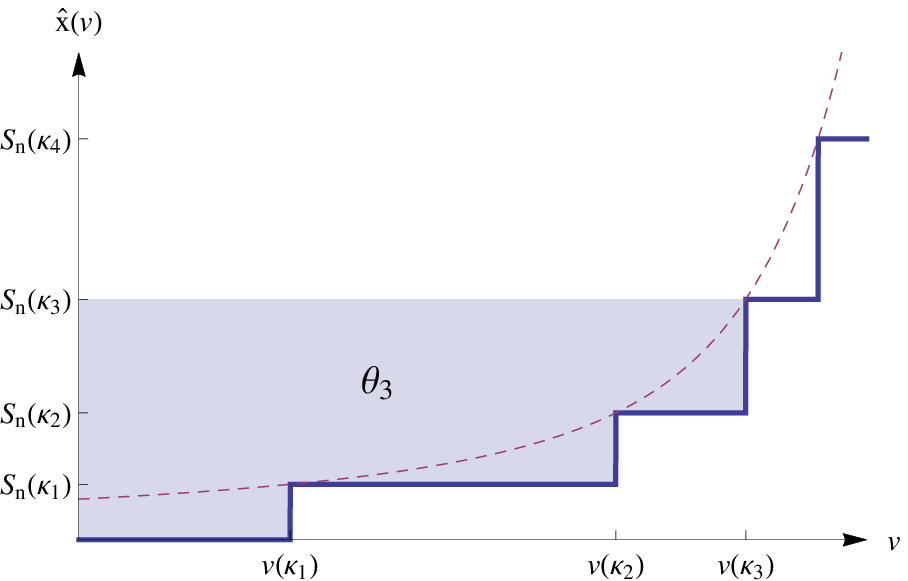}}
\caption{Equilibrium bid function (left) and relation between interim status allocation and contribution thresholds at equilibrium (right).}\label{fig:interim-characterization}
\end{figure}
Equivalently, the above set of conditions can be re-written as:
\begin{align}\label{eqn:value-to-effort}
\forall t\in [1,\ldots,p]:~& \val(\tq{t})\cdot \left(\statn{n}{\tq{t}}-\statn{n}{\tq{t-1}}\right)=\tb{t}-\tb{t-1}
\end{align} 
This set of equalities has an intuitive interpretation as saying that the players with quantiles at the boundary of two badges should be indifferent between getting either of the two badges. Additionally, to ensure that the latter is an equilibrium we also need to make sure that if $p<m$, then the highest ability player doesn't prefer being on badge $p+1$ alone, rather than being on badge $p$:
\begin{align}\label{eqn:value-to-effort-2}
\text{if } p<m:~& \val(0)\cdot \left(\statn{n}{0}-\statn{n}{\tq{p}}\right)\leq \tb{p+1}-\tb{p}
\end{align}
To show uniqueness of the symmetric equilibrium we simply need to show that the above set of conditions have a unique solution.

\begin{lemma}\label{lem:uniqueness} There exists a unique $p\leq m$ and a unique vector $\tqvec=(\tq{1},\ldots,\tq{p})$, that satisfies the system of Equations \eqref{eqn:value-to-effort} and \eqref{eqn:value-to-effort-2}.
\end{lemma}
\begin{proof}
Given a profile of badge thresholds $\tbvec$, we show recursively that there is a unique set of quantile thresholds which satisfies the set of equations. For $t=1$, Equation \eqref{eqn:value-to-effort} becomes: $\tb{1} = \val(\tq{1})\cdot \statn{n}{\tq{1}}$. Observe that $\val(0)\cdot \statn{n}{0}=\bar{\val}\cdot \stat{0}$, $\val(1)\cdot \statn{n}{1}=0$, $\val(x)\cdot \statn{n}{x}$ is continuous decreasing. If $\tb{1}<\bar{\val}\cdot \stat{0}$ then a unique solution exists (recall that $\bar{\val}$ is the upper bound of the ability distribution). Otherwise, no player is willing to bid as high as $\tb{1}$ (or any $\tb{t}$ for $t>1$) and the recursion stops with $p=0$. Subsequently, find the solution $\tv{2}$ to the equation: $\tb{2}-\tb{1}= \val(\tq{2}) (\statn{n}{\tq{2}} - \statn{n}{\tq{1}})$. For similar reason, either a unique such solution exists or no player (not even a player with ability $\bar{\val}$) is willing to bid $\tb{2}$ rather than bid $\tb{1}$ and we can stop the recursion, setting $p=1$. Then solve for $\tq{3},\ldots, \tq{p}$ in the same way.
\end{proof}
The latter Lemma completes the proof of the uniqueness and characterization of the equilibrium. The inverse part of the theorem, is trivial based on the previous discussion. 
\end{proofof}

%% file: appendix-single.tex
\begin*{ THEOREM \ref{cor:median-linear}.}
\ \ \em When the status function $\stat{\cdot}$ is concave, the expected total contribution of an absolute threshold mechanism with a single quantile threshold $\tq{}=\frac{1}{2}$ is a $4$-approximation to the expected total contribution of the optimal badge mechanism.
\end*{}
\\ \\
\begin{proofof}{Theorem \ref{cor:median-linear}}
Setting a quantile threshold of $\tq{} = 1/2$, by Equation \eqref{singleRev}, yields an expected user contribution of 
$$\apx = R\left(\frac{1}{2}\right)\cdot \statn{n}{\frac{1}{2}}.$$ Our first step is to upper bound the expected user contribution of the optimal mechanism, and then to show that $R(1/2)\cdot \statn{n}{1/2}$ is a 4-approximation of this upper bound. By symmetry, an approximation of the expected user contribution, implies the same approximation to the expected total contribution.

By Equation \ref{eqn:optimal-user-contr} and using the fact that $\statn{n}{\quant}\leq \statn{n}{0}$ and that 
$R'(\quant)\geq 0$ for any $\quant\leq \tqstar$:
\begin{equation*}\label{eqn:linear-optimal-revenue}
\opt =  \int_0^{\tqstar} R'(\quant)\cdot \statn{n}\quant d\quant \leq \statn{n}{0}\int_{0}^{\tqstar} R'(\quant)d\quant \leq 
\statn{n}{0}\cdot R(\tqstar)
\end{equation*}
Now we prove $ R\left(\frac{1}{2}\right)\cdot \statn{n}{\frac{1}{2}} \geq \frac{1}{4}R(\tqstar) \cdot \statn{n}{0} $.  By the concavity of $\statn{n}{\cdot} $ and the fact that $\statn{n}{1}=0$, we get 
\[ \statn{n}{\frac{1}{2}} \geq \frac{1}{2} \statn{n}0 + \frac{1}{2}\statn{n}1 = \frac{1}{2} \statn{n}0 \]
From the concavity of the revenue function $R(\cdot)$ and Jensen's inequality, we get
\[ R\left(\frac{1}{2}\right) \geq \int_0^1 R(\quant)d\quant \geq \frac{1}{2} R(\tqstar) \]
where the first inequality follows from Jensen's inequality and the last inequality follows from concavity and the fact that $R(0)=R(1)=0$ and $R(\tqstar)$ is the maximum. \footnote{The latter is the same property of regular distributions
employed for proving Bulow-Klemperer's result that the revenue of the optimal single item auction with one bidder yields at most the revenue of a second price auction with two i.i.d bidders. See Figure 1 of \cite{Dhangwatnotai2010} or Figure 5.1 in \cite{Hartline2012}} Putting this all together
\[ \apx = R\left(\frac{1}{2}\right) \cdot \statn{n}{\frac{1}{2}} \geq \frac{1}{4} R(\tqstar)\cdot \statn{n}{0} \geq \frac{1}{4} \opt \]
\end{proofof}

\begin*{THEOREM \ref{thm:better-apx}}
\ \ \em When the status value function $\stat{\cdot}$ is concave, the expected total contribution of an absolute threshold mechanism with a single quantile threshold $\tq{}=\min\{\tqstar,1/2\}$ is a 3-approximation to the expected total contribution of the optimal badge mechanism. Furthermore, when the status value function is linear, it is a 2-approximation.
\end*{}
\\ \\
\begin{proofof}{Theorem \ref{thm:better-apx}} Throughout the proof we will denote with $\opt$ the expected user contribution of the optimal mechanism and with $\apx$ the expected user contribution of the single absolute badge mechanism
with quantile threshold $\tq{}=\min\{\tqstar,1/2\}$. We remind that $\opt$ and $\apx$ are characterized by Equations \eqref{eqn:optimal-user-contr} and \eqref{singleRev} respectively.

\paragraph{Linear Case} We begin by proving the case where the status value function $\stat{\cdot}$ is linear. 

If $\tqstar\leq 1/2$, then $\tq{}=\tqstar$ and:
$$\opt=\int_0^{\tqstar}R'(\quant)\cdot(1-\quant)d\quant\leq  R(\tqstar) \leq  R(\tqstar)\cdot 2\cdot (1-\tqstar)=2\cdot \apx.$$ 

If $\tqstar> 1/2$, then $\tq{}=1/2$. Consider the concave curve defined as 
\begin{equation*}
\hat{R}(\quant)=\begin{cases}R(\quant) & \quant\in[0,\tqstar)\\ R(\tqstar) & \quant \in [\tqstar,1]\end{cases}
\end{equation*} 
By Equation \eqref{eqn:optimal-user-contr}, observe that $\opt= \int_{0}^{1}\hat{R}(\quant)d\quant$. By concavity of $\hat{R}(\quant)$
and applying Jensen's inequality we get that:
\begin{equation}
\int_{0}^{1}\hat{R}(\quant)d\quant \leq \hat{R}\left(\frac{1}{2}\right) =R\left(\frac{1}{2}\right)
\end{equation}
where in the last equality we used the fact that $\tqstar>1/2$ and thereby $\hat{R}(1/2)=R(1/2)$. Thus we get that:
$\opt\leq R(1/2)$. A single badge mechanism with quantile threshold at $1/2$ gets revenue $\apx =\frac{1}{2}R(1/2)\geq \frac{1}{2}\opt$.

Thus a single badge mechanism with quantile $\tq{}=\min\{\tqstar,1/2\}$ yields a
$2$-approximation to the total contribution of the optimal mechanism in any case.

\paragraph{Concave Case} We now consider the case where the status value function $\stat{\cdot}$ is concave.


If $\tqstar\leq 1/2$, then $\tq{}=\tqstar$. By the concavity of $\statn{n}{\cdot}$ and the fact that $\statn{n}{1}=0$: $\statn{n}{\tqstar}\geq (1-\tqstar)\cdot \statn{n}{0} \geq \frac{1}{2}\statn{n}{0}$. Thus: 
$$\apx=R(\tqstar)\cdot \statn{n}{\tqstar}\geq R(\tqstar)\cdot \frac{1}{2}\cdot \statn{n}{0}\geq \frac{1}{2} \opt.$$
The fact that $\opt \leq R(\tqstar)\cdot \statn{n}{0}$, comes from the simple fact that $\statn{n}{q}\leq \statn{n}{0}$ and by replacing it in Equation \eqref{eqn:optimal-user-contr}.

If $\tqstar>1/2$ then $\tq{}=1/2$. We will use the following simple facts: 
\begin{enumerate}
\item Since $R(\quant)$ is increasing concave for any $\quant\in [0,\tqstar]$ and $R(0)=0$, then for any
$$t\in [1/2,\tqstar]:~R'(t)\leq R'(1/2)\leq \frac{R(1/2)}{1/2}=2\cdot R(1/2),$$ 
\item Since $\statn{n}{\quant}$ is a decreasing concave function and $\statn{n}{1}=0$, then for any 
$$t\in[0,1/2]:~ -\statnpr{n}{t} \leq -\statnpr{n}{1/2} \leq \frac{\statn{n}{1/2}}{1-1/2}=2\cdot \statn{n}{1/2}.$$
\end{enumerate}
Using these properties and an application of integration by parts, we can upper bound the expected user contribution of the optimal badge mechanism:
\begin{align*}
\opt =~&  \int_0^{\tqstar} R'(\quant)\cdot \statn{n}{\quant}d\quant =  \int_0^{1/2} R'(\quant)\cdot \statn{n}{\quant}d\quant+ \int_{1/2}^{\tqstar} R'(\quant)\cdot \statn{n}{\quant}d\quant\\
=~& R\left(\frac{1}{2}\right)\cdot \statn{n}{\frac{1}{2}} + \int_0^{1/2} R(\quant)\cdot(-\statnpr{n}{\quant})d\quant+ \int_{1/2}^{\tqstar} R'(\quant)\cdot \statn{n}{\quant}d\quant\\
\leq~&  R\left(\frac{1}{2}\right)\cdot \statn{n}{\frac{1}{2}} + 2\cdot \statn{n}{\frac{1}{2}}\cdot \int_0^{1/2} R(\quant)d\quant+ 2\cdot R\left(\frac{1}{2}\right)\cdot \int_{1/2}^{\tqstar} \statn{n}{\quant}d\quant\\
\leq~& R\left(\frac{1}{2}\right)\cdot \statn{n}{\frac{1}{2}} + 2\cdot \statn{n}{\frac{1}{2}}\cdot \int_0^{1/2} R\left(\frac{1}{2}\right)d\quant+ 2\cdot R\left(\frac{1}{2}\right)\cdot \int_{1/2}^{\tqstar} \statn{n}{\frac{1}{2}}d\quant\\
\leq~& 3\cdot R\left(\frac{1}{2}\right)\cdot \statn{n}{\frac{1}{2}} = 3\cdot \apx
\end{align*}

Therefore, in any case, setting $\tq{}=\min\{1/2,\tqstar\}$, yields a $3$-approximation to the optimal revenue.
\end{proofof}

\begin{example}[Tight lower bound for median badge mechanism]\label{ex:four-lower-bound} Let the status value function be $\stat{t}=1-t$. We will show that when the distribution of abilities has a long tail then the median badge mechanism is at most a $4$-approximation to the optimal badge mechanism. In this example, only the top ability make significant contributions in the optimal mechanism. The median badge mechanism sets too low of a contribution threshold to achieve a better approximation.


Specifically, suppose that the distribution of abilities has a cumulative density function of $$F(\val)=\frac{H+1}{H}\frac{\val}{\val+1}$$ and support $[0,H]$, and consider the limit as $H\rightarrow\infty$.\footnote{In this limit, the distribution of abilities converges to a translation of what is called the \emph{equal revenue distribution}.}. The revenue function $R(q)$ of such an ability distribution is 
\begin{equation}
R(\quant)=\quant\cdot \val(\quant) = \quant\cdot F^{-1}(1-\quant) = \quant\frac{1-\quant}{\frac{1}{H}+\quant}\overset{H\rightarrow \infty}{\rightarrow} 1-\quant.
\end{equation}
The corresponding monopoly quantile is $\tqstar=\frac{\sqrt{1+H}-1}{H}\rightarrow 0$ as $H\rightarrow \infty$, and so, by Equation \eqref{eqn:linear-contribution}, the expected user contribution in the optimal mechanism is $(1-\tqstar)R(\tqstar)+\int_{0}^{\tqstar}R(q)dq\rightarrow1$, whereas in the median badge mechanism, by Equation \eqref{singleRev}, it is $(1-1/2)R(1/2)\rightarrow1/4$.

Graphically, our example corresponds to the case where the revenue curve converges to a "left triangle", of the form $R(\quant)=1-\quant$. In that case, the expected user contribution of the optimal badge mechanism corresponds to the rectangle of height and length $1$, while of the median badge mechanism corresponds to the rectangle of height and length $1/2$ (see Figure~\ref{fig:linear}).

\end{example}

\begin{figure}[htpb]
\centering
\subfigure{\includegraphics[scale=.75]{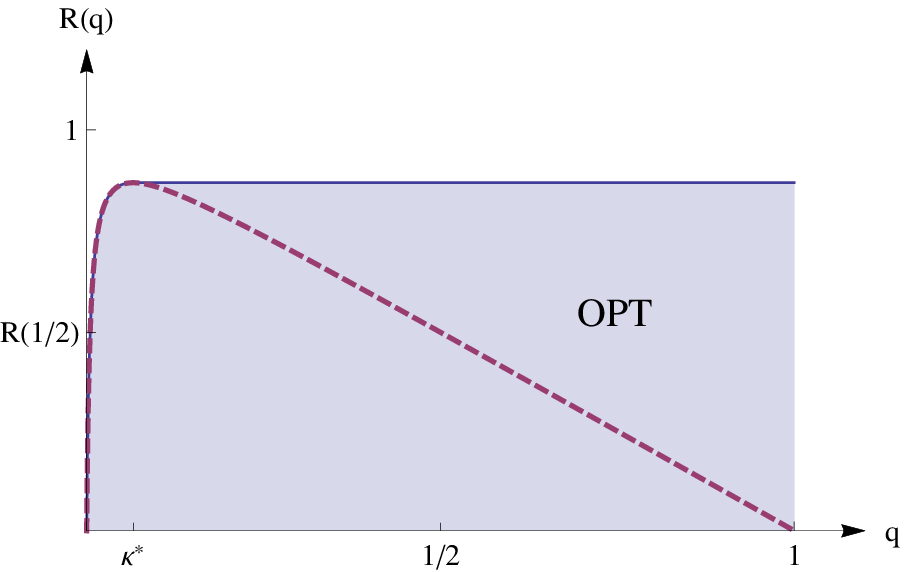}}
\subfigure{\includegraphics[scale=.75]{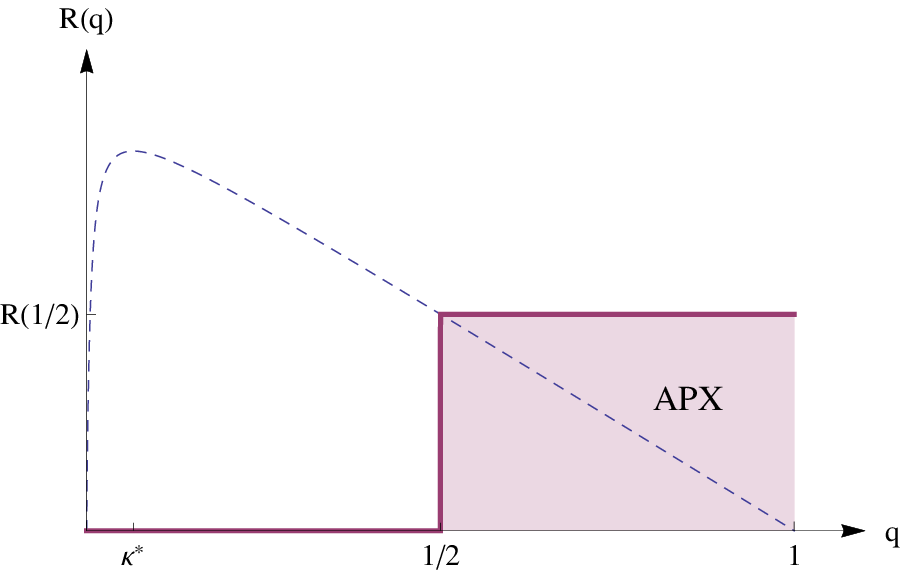}}
\caption{Contributions of the optimal and median badge mechanism for ability distributions of the form $\cdf[\val]=\frac{H+1}{H}\frac{v}{v+1}$, respectively.} \label{fig:linear}
\end{figure}

%% file: appendix-many.tex

The examples from section \ref{sec:lowerbound} demonstrate the limitations of mechanisms that use a single threshold. With a single threshold, no mechanism can achieve better than a 2-approximation for concave or linear status. Furthermore, no finite approximation is possible for convex status valuations with any constant number of thresholds. We now characterize the approximate optimality of mechanisms with $m > 1$ thresholds. For concave status, the contributions generated by badge mechanism with $m$ badges quickly approaches the contributions of the optimal mechanism.

\begin{theorem}\label{thm:many-badge-concave}
If the social status function $\stat{\cdot}$ is concave, then a badge mechanism with $m$ badges, can achieve a $\frac{m}{m-2}$-approximation to the total contribution of the optimal mechanism. 
\end{theorem}
\begin{proof}
 Throughout the proof we will denote with $\opt$ the expected user contribution of the optimal mechanism and with $\apx$ the expected user contribution of an absolute threshold mechanism characterized by a vector of quantile thresholds $\tqvec=(\tq{1},\ldots,\tq{m})$. We remind that $\opt$ is given by Equation \eqref{eqn:optimal-user-contr} and will first provide a characterization of $\apx$ as a function of the quantile threshold vector.

Observe that by the form of the equilibrium described in Theorem \ref{thm:uniqueness}, we get that the interim status allocation
of a player in the absolute threshold mechanism is:
\begin{eqnarray*}
\xq{\quant} = 
\begin{cases}
0 & \quant > \tq{1} \\
\statn{n}{\tq{t}} & \quant \in (\tq{t+1},\tq{t}] \\
\statn{n}{\tq{m}} & \quant \leq \tq{m} 
\end{cases}
\end{eqnarray*}
Thus by applying the generic expected user contribution characterization of Lemma \ref{lem:myerson-opt}, we get:
\begin{equation}\label{eqn:badge-mechanism-rev-1}
\apx = \int_{0}^{1} R'(\quant)\cdot \xq{\quant}d\quant= R(\tqstar)\cdot \statn{n}{\tqstar}+\sum_{t=2}^{m}R(\tq{t})(\statn{n}{\tq{t}}-\statn{n}{\tq{t-1}})
\end{equation}

Consider the badge mechanism with quantile thresholds defined so that they satisfy the following conditions: 
\begin{align*}
\tq{1}=~&\tqstar\\
\forall t=\{2,\ldots,m\}: \statn{n}{\tq{t}}=~&\statn{n}{\tqstar}+(t-1)\cdot \Delta x,
\end{align*} 
where 
$$\Delta x=\frac{\statn{n}{0}-\statn{n}{\tqstar}}{m}.$$ 
This implies that: $\statn{n}{\tq{t}}-\statn{n}{\tq{t-1}}=\Delta x$ for any $t\in[2,m]$ and $\statn{n}{0}-\statn{n}{\tq{m}}=\Delta x$. For convenience we will
denote with $\tq{m+1}=0$ and $\tq{0}=1$.

By Equation \eqref{eqn:badge-mechanism-rev-1}, the expected user contribution of the above absolute threshold mechanism is:
\begin{align*}
\apx= ~& R(\tqstar)\cdot \statn{n}{\tqstar}+\sum_{t=2}^{m}R(\tq{t})(\statn{n}{\tq{t}}-\statn{n}{\tq{t-1}})
 =~ R(\tqstar)\cdot \statn{n}{\tqstar}+\sum_{t=2}^{m} R(\tq{t})\cdot\Delta x
\end{align*}

On the other hand the expected user contribution of the optimal mechanism can be lower bounded by applying integration by parts to Equation \eqref{eqn:optimal-user-contr} and using the monotonicity of the revenue function in the 
region $[0,\tqstar]$:
\begin{align*}
\opt =~& \int_{0}^{\tqstar} R'(\quant)\cdot \statn{n}{\quant}d\quant = R(\tqstar)\cdot \statn{n}{\tqstar}-\int_{0}^{\tqstar}R(\quant)\statnpr{n}{\quant}dq\\
=~&  R(\tqstar)\cdot \statn{n}{\tqstar}-\sum_{t=1}^{m}\int_{\tq{t+1}}^{\tq{t}}R(\quant)\cdot \statnpr{n}{\quant}dq\\
\leq~&R(\tqstar)\cdot \statn{n}{\tqstar}-\sum_{t=1}^{m}R(\tq{t})\int_{\tq{t+1}}^{\tq{t}}\statnpr{n}{\quant}dq
=~ R(\tqstar)\cdot \statn{n}{\tqstar}+\sum_{t=1}^{m}R(\tq{t})\cdot\Delta x
\end{align*}
Thus we get that:
\begin{equation}\label{eqn:diff-upper-bound}
\opt-\apx \leq R(\tqstar)\cdot \Delta x = R(\tqstar)\cdot \frac{\statn{n}{0}-\statn{n}{\tqstar}}{m} \leq R(\tqstar)\cdot \frac{\statn{n}{0}}{m}
\end{equation}
We will now show that $\opt \geq \frac{1}{2}  R(\tqstar)\cdot \statn{n}{0}$.
Since, the revenue function $R(\quant)$ is concave and $R(0)=0$, for any 
$$\quant\in [0,\tqstar]: R(\quant)\geq \frac{R(\tqstar)}{\tqstar}\quant.$$ 
Thus:
\begin{align*}
\opt=~&R(\tqstar)\cdot \statn{n}{\tqstar}-\int_{0}^{\tqstar}R(\quant)\cdot \statnpr{n}{\quant}d\quant
 \geq~ R(\tqstar)\cdot \statn{n}{\tqstar}-\frac{R(\tqstar)}{\tqstar}\int_{0}^{\tqstar}\quant\cdot \statnpr{n}{\quant}d\quant\\
=~&R(\tqstar)\cdot \statn{n}{\tqstar}-\frac{R(\tqstar)}{\tqstar}\left(\tqstar\cdot\statn{n}{\tqstar}-\int_{0}^{\tqstar}\statn{n}{\quant}d\quant\right)
=~ \frac{R(\tqstar)}{\tqstar}\int_{0}^{\tqstar}\statn{n}{\quant}d\quant 
\end{align*}
Since, $\statn{n}{\quant}$ is a non-negative concave decreasing function, we have that:
\begin{equation}
\int_0^{\tqstar} \statn{n}{\quant}d\quant \geq \frac{\statn{n}{0}+\statn{n}{\tqstar}}{2}\tqstar \geq \frac{\statn{n}{0}}{2}\tqstar 
\end{equation}
Thus we get:
\begin{equation}\label{eqn:optimal-lower-bound}
\opt \geq  \frac{1}{2}\cdot  R(\tqstar)\cdot \statn{n}{0} 
\end{equation}
Combining Equations \eqref{eqn:diff-upper-bound} and \eqref{eqn:optimal-lower-bound}:
\begin{equation}
\opt  - \apx \leq\frac{2}{m} \opt
\end{equation}
which yields the theorem.
\end{proof}

For convex valuations, we show that a 4-approximation is possible if the mechanism uses a number of thresholds that is logarithmic in a natural parameter of the status valuation function: $H = \frac{\stat{0}}{\stat{\frac{1}{2}}}$, i.e. the ratio of the status value of the highest ranked user to the status value of a median-ranked user. The next theorem shows that $\log(H)$ badges are sufficient for achieving a constant approximation. For a large class of natural status valuations, this ratio $H$ will be on the order $O(n)$, the number of users. For example, if we use the status value function $\stat{t} =\frac{n}{(n-1)t+1}-1$ from example \ref{convexLowerBound}, then $\stat{0} = n-1$ and $\stat{\frac{1}{2}} = \frac{n-1}{n+1}$ and hence a badge mechanism with $\log(H) = \log(n+1)$ thresholds yields a 4-approximation to the optimal mechanism. Thus the number of necessary thresholds grows slowly with the number of users. 

\begin{theorem}\label{thm:many-badge-convex}
Let $\stat{\cdot}$ be any convex status function and let $\lambda = \min\{ \tqstar ,\frac{1}{2}\}$. The badge mechanism with quantile thresholds $\tqvec=(\tq{1},\ldots,\tq{m})$,  where $m= \log\left(\frac{\stat{0}}{\statn{n}{\lambda}}\right)\leq \log(H)$ and $\tqvec$ satisfies the following:
\begin{align*}
\statn{n}{\tq{m}}= & ~ \frac{\stat{0}}{2}\\
\statn{n}{\tq{t}} = & ~ \frac{ \statn{n}{\tq{t+1}}}{2} = \frac{\stat{0}}{2^{m-t+1}} \qquad  \forall t\in[2, m-1] \\
\tq{1}= & ~ \lambda
\end{align*}
achieves a $4$-approximation to the total contribution of the optimal mechanism. 
\end{theorem}
\begin{proof}
Let $\opt$ be the expected user contribution of the optimal badge mechanism and $\apx$ the expected user contribution of the described absolute threshold mechanism. We will show that the interim status allocation of a user with quantile $\quant \leq \lambda= \min\left\{\tqstar,\frac{1}{2}\right\}$ in the described absolute threshold mechanism is at least half of his interim status allocation in the optimal mechanism. This property does not hold for $\quant \geq \lambda$ but we prove that the optimal mechanism generates at most half of its contributions from users with $q > \lambda$. The 4-approximation result then follows, by the virtual surplus
characterization of total contribution.

By proposition \ref{prop:badge-interim}, we know that any $\quant \leq \tqstar$ has an interim status allocation of $\statn{n}{q}$ in the optimal mechanism. We now prove that for any $\quant \leq \lambda$, the interim status allocation in this mechanism is at least half of the interim status allocation of the optimal mechanism. 

First, consider any user with $\quant < \tq{m}$. By definition of this mechanism:
\[ \xq{\quant} = \statn{n}{\tq{m}} = \frac{ \stat{0} }{2} \geq \frac{ \stat{\quant}}{2} \]

Next, consider any user with $\quant \in (\tq{t+1},\tq{t}]$ for $t \in \{2,...,m-1\}$. 
\[ \xq{\quant} = \statn{n}{\tq{t}} = \frac{\statn{n}{\tq{t+1}}}{2} \geq \frac{\statn{n}{\quant}}{2} \]  

Finally, for any $\quant \in (\tq{2}, \lambda]$. 
\[ \xq{\quant} = \statn{n}{\lambda} = \frac{\stat{0}}{2^m} = \frac{1}{2} \cdot \left( \frac{\stat{0}}{2^{m-1}} \right) = \frac{1}{2} \cdot \statn{n}{\tq{2}} \geq \frac{\statn{n}{\quant}}{2} \]


Since this mechanism assigns non-zero status value only to users with $q\leq \lambda \leq \tqstar$ and $\xq{\quant} \geq \frac{\statn{n}{\quant}}{2}$ we get:
\begin{equation*}
\apx =  \int_0^{\lambda} R'(\quant) \cdot \xq{\quant}d\quant\geq \frac{1}{2}\cdot \int_0^{\lambda} R'(\quant) \cdot \statn{n}{\quant}d\quant
\end{equation*} 

\noindent On the other hand the optimal mechanism achieves expected user contribution:
\begin{equation*}
\opt = \int_0^{\tqstar} R'(\quant)\cdot \statn{n}{\quant} d\quant 
\end{equation*}
If $\tqstar\leq 1/2$, then we get $2 \cdot \apx \geq \opt$.

If $\tqstar >1/2$, then:
\begin{equation*}
\opt = \int_0^{1/2}R'(\quant)\cdot \statn{n}{\quant} d\quant +  \int_{1/2}^{\tqstar}R'(\quant)\cdot \statn{n}{\quant}d\quant
\end{equation*}
$R'(\cdot)$ and $\statn{n}{\cdot}$ are both non-increasing functions of the quantile, so the following inequality holds: 
\begin{equation*}
\int_0^{1/2}R'(\quant)\cdot \statn{n}{\quant} d\quant \geq \int_{1/2}^{\tqstar}R'(\quant)\cdot \statn{n}{\quant}d\quant
\end{equation*} 
and thereby:
\begin{equation*}
\opt \leq 2\cdot  \int_0^{1/2}R'(\quant)\cdot \statn{n}{\quant} d\quant \leq 4 \cdot \apx
\end{equation*}
\end{proof}

%% file: appendix-complete-ranking.tex
\begin*{}{THEOREM \ref{thm:leaderboard-apx}}
\ \ \em For any convex status function, the leaderboard mechanism is a $2$-approximation to the optimal mechanism.
\end*{}
\\ \\
\begin{proofof}{Theorem \ref{thm:leaderboard-apx}}
Following similar reasoning as in Theorem \ref{thm:optimal}, we can show that such a complete relative ranking mechanism, will have a unique symmetric monotone equilibrium. Under such an equilibrium a player with quantily $\quant$, will be ranked below all player with lower quantile and above all players with higher quantile. Thus the interim status allocation for a player with quantile $\quant$ is $\statn{n}{\quant}$, as defined in Equation \eqref{eqn:expectedstatus}. Thus, the expected user contribution will be:
\begin{equation}
\apx = \int_0^1 R'(\quant)\cdot \statn{n}{\quant}d\quant
\end{equation}
whereas the optimal mechanism induces an interim status allocation of $\statn{n}{\quant}$ only for players with $\quant \leq \tqstar$ and 0 otherwise, yielding an expected user contribution of $\opt$ as given in Equation \eqref{eqn:optimal-user-contr}. We will show that the convexity of the status function and the regularity of the ability distribution imply that $\opt\leq 2\cdot \apx$.

By convexity of $\statn{n}{\cdot}$, $\statn{n}{t\cdot k+(1-t)\cdot1}\leq t\statn{n}{k}+(1-t)\statn{n}{1}$.  Instantiating this for $t=(1-q)/(1-\tqstar)$ and recalling that $\statn{n}{1}=0$, we get that for any $\quant\geq \tqstar: \statn{n}{\quant}\leq \statn{n}{\tqstar}\cdot (1-\quant)/(1-\tqstar)$. Since by definition $R'(\quant)$ is non-positive for any $\quant\geq \tqstar$, we can lower bound the negative part of $\apx$ as follows:
\begin{align*}
\int_{\tqstar}^{1}R'(\quant)\cdot \statn{n}{\quant}d\quant\geq 
~&\int_{\tqstar}^{1} R'(\quant) \frac{\statn{n}{\tqstar}}{1-\tqstar}(1-\quant)d\quant\\
=~&  \frac{\statn{n}{\tqstar}}{1-\tqstar} \int_{\tqstar}^{1}R'(\quant)(1-\quant)d\quant\\
=~& \frac{\statn{n}{\tqstar}}{1-\tqstar}\left(-R(\tqstar)(1-\tqstar)+\int_{\tqstar}^{1}R(\quant)d\quant\right)\\
\geq~& -\frac{1}{2}\cdot\statn{n}{\tqstar}\cdot R(\tqstar)
\end{align*}
where the last inequality follows since, by the fact that $R(\quant)$ is concave non-increasing in the region $[\tqstar,1]$ and the assumption that $R(1)=0$, we have that $\int_{\tqstar}^{1} R(\quant)d\quant \geq \frac{1}{2} R(\tqstar)\cdot(1-\tqstar)$.
We can now lower bound \apx, using integration by parts:
\begin{align*}
\apx \geq~& \int_0^{\tqstar} R'(\quant)\cdot \statn{n}{\quant}d\quant - \frac{1}{2}\cdot\statn{n}{\tqstar}\cdot R(\tqstar)\\
=~& R(\tqstar)\cdot \statn{n}{\tqstar} - \int_0^{\tqstar}R(\quant)\cdot\statnpr{n}{\quant}d\quant - \frac{1}{2}\cdot\statn{n}{\tqstar}\cdot R(\tqstar)\\
=~& \frac{1}{2}\cdot R(\tqstar)\cdot \statn{n}{\tqstar} - \int_0^{\tqstar}R(\quant)\cdot\statnpr{n}{\quant}d\quant \\
\geq~& \frac{1}{2}\cdot R(\tqstar)\cdot \statn{n}{\tqstar} - \frac{1}{2}\int_0^{\tqstar}R(\quant)\cdot\statnpr{n}{\quant}d\quant \\
=~& \frac{1}{2} \int_0^{\tqstar} R'(\quant)\cdot \statn{n}{\quant}d\quant =~ \frac{1}{2} \opt
\end{align*}
where in the last inequality we also used the fact that $\statnpr{n}{\quant}\leq 0$.

\end{proofof}

%% file: appendix-tie-breaking.tex
\begin*{}{THEOREM \ref{thm:median-tie-breaking}}
\ \ \em For any $\beta\in [0,1]$, the \emph{median badge mechanism} achieves a $4$ approximation to the contribution of the optimal mechanism. The \emph{median badge mechanism} is implemented at the unique symmetric monotone equilibrium of the absolute threshold mechanism defined by a contribution threshold of $\frac{v(1/2)}{2}$, when $\beta\geq 1/2$ and at some equilibrium for $\beta<1/2$.
\end*{}
\\ \\
\begin{proofof}{Theorem \ref{thm:median-tie-breaking}}
Let $\opt$ be the expected user contribution of the optimal badge mechanism and $\apx$ the expected user contribution of the  absolute threshold mechanism with contribution threshold $\tb{}=\frac{v(1/2)}{2}$. First we analyze the equilibrium induced by setting a single badge threshold of $\tb{}$ (observe that our equilibrium characterization in Theorem \ref{thm:uniqueness} depends on the fact that we used $\beta=1$). We focus on symmetric monotone equilibria, and thereby the equilibrium of such a mechanism
is characterized by a quantile threshold $\tq{}$, such that all users with quantile $\quant\leq \tq{}$, submit $\tb{}$, and all users with $\quant>\tq{}$, submit $0$. 

By the indifference of the user at the boundary quantile $\tq{}$, it must be that:
\begin{equation}\label{eqn:indiff-beta}
v(\tq{})\left(1-\beta\cdot \tq{}\right) -\theta= v(\tq{})\left(1-\tq{}-\beta\cdot(1-\tq{})\right) \implies v(\tq{})\left(\tq{}\cdot(1-2\beta)+\beta\right)=\theta
\end{equation}
Observe that if $\beta\geq1/2$ then the left hand side is monotone-decreasing in $\tq{}$ and therefore the equation
has at most one solution. For $\beta<1/2$, there might be multiple solutions and thereby multiple equilibria.
Consider setting $\tb{}=\frac{v(1/2)}{2}$. Then Equation \eqref{eqn:indiff-beta} has solution $\tq{}=1/2$ independent of $\beta$ and if $\beta\geq 1/2$, it is the unique solution. 

The expected user contribution achieved by the median quantile threshold equilibrium is:
\begin{equation}
\apx = \tq{} \cdot \tb{} = \frac{1}{2} \cdot \frac{v(1/2)}{2} =\frac{R(1/2)}{2} \geq \frac{R(\tqstar)}{4} \geq \frac{1}{4} \opt
\end{equation}
Where we used the fact that $R(1/2)\geq R(\tqstar)/2$, by the regularity of the distribution. Moreover, we used the upper bound on the optimal mechanism of $R(\tqstar)$. This fact can be seen as follows: since $S(t_e,t_g)\leq 1$, the interim status allocation of any player is at most $1$. No matter what the optimal mechanism is, each user's $i$ expected contribution can be upper bounded by:
\begin{align*}
\opt_i=\int_0^1 R'(\quant)\cdot \hat{x}_i(\quant)d\quant =~& \int_0^{\tqstar} R'(\quant)\cdot \hat{x}_i(\quant)d\quant + \int_{\tqstar}^1 R'(\quant)\cdot \hat{x}_i(\quant)d\quant \leq R(\tqstar)
\end{align*}
Since, the first integral in the left hand side of the last inequality is bounded above by $R(\tqstar)$ (since $R'(\cdot)$ is non-negative), while the second integral is bounded above by $0$ (since $R'(\cdot)$ is non-positive). 
\end{proofof}

\begin*{}{THEOREM \ref{thm:complete-tie-breaking}}
\ \ \em For any $\beta\in [0,1]$, the badge mechanism that assigns a distinct badge to each user in decreasing order of contribution, is always a $2$-approximation to the total contribution of the optimal mechanism.
\end*{}
\\ \\
\begin{proofof}{Theorem \ref{thm:complete-tie-breaking}}
Let $\opt$ be the expected user contribution of the optimal badge mechanism and $\apx$ the expected user contribution of the
complete relative ranking mechanism. As argued in the proof of Theorem \ref{thm:leaderboard-apx}, such a complete ranking badge mechanism has a unique equilibrium at which users are ranked in decreasing order of quantile. Therefore, the interim status allocation of each user is simply: 
$\xq{\quant}=S_n(\quant)=1-\quant$ and the expected user contribution of the mechanism is:
$$\apx=-\int_{0}^{1}R(\quant)\xqpr{\quant}d\quant = \int_0^{1} R(\quant) d\quant \geq \frac{1}{2}\cdot R(\tqstar)\geq \frac{1}{2}\opt,$$
where the second to last inequality follows by the same argument as in Theorem \ref{cor:median-linear} and the last inequality follows from the same argument as in the proof of Theorem \ref{thm:median-tie-breaking}. 
\end{proofof}

\subsection{Structure of Optimal Badge Mechanism under Different Tie-Breaking Rules}
Apart from the case of $\beta=1$, discussed in Section \ref{sec:optimal}, we show that the surplus-maximizing allocation has a nice structure, when $\beta$ takes values $1/2$ or $0$.
Specifically, for the case of $\beta=1/2$, we show that assigning all users a distinct badge in decreasing order of value is optimal. Hence, a complete ranking mechanism 
that assigns a distinct badge in decreasing order of bid (without any contribution threshold) is optimal. Observe, that for $\beta=1$, the optimal mechanism also had a contribution
threshold. For $\beta=0$, we show that the optimal mechanism groups together all users below the monopoly quantile in the same and highest badge, and then assigns
a distinct decreasing badge to all users above the monopoly quantile. Such a mechanism can also be implemented at the unique symmetric monotone equilibrium of an all-pay ranking
mechanism, where user's above some contribution threshold $\theta$ are all assigned the top badge and all remaining users are ranked in decreasing order of contribution. For other values of $\beta$, ex-post maximization of virtual surplus, can be very complex and dependent on the exact instantiation of the virtual surplus profile of users. 

\begin{theorem}\label{thm:half-ties} 
If $\beta=1/2$, then the virtual surplus maximizing allocation assigns all users a distinct badge in decreasing order of quantile (increasing order of ability). Such an allocation is implemented at the unique equilibrium of a leaderboard mechanism.
\end{theorem}
\begin{proof}
To argue the first part of the theorem, we argue that for any instantiation of user quantiles, the virtual surplus maximizing allocation is to assign a distinct badge to all users in decreasing order of quantile. Similar to Lemma \ref{lem:optimal-surplus}, it is easy to see that the badge should be monotone non-decreasing in the virtual ability, since if 
$R'(\quant[i])>R'(\quant[j])$ and $r_i<r_j$ then we can increase virtual surplus by swapping the badge of user $i$ and user $j$. Additionally, if for some set of virtual abilities $R'(\quant[i])\geq R'(\quant[i+1])\geq \ldots R'(\quant[i+k])$ we have $r_i=\ldots = r_{i+k}$ then by discriminating user $i$ to a higher badge then we can argue that the virtual surplus will increase: More concretely, for any $j> i$ we can set $r_j'=r_j+1$. The status value of all users $j>i+k$ remains the same, while the satisfaction of
each user $j\in[i+1,i+k]$ will decrease by $\frac{1}{2(n-1)}$. The status value of user $i$, will increase by $\frac{k-1}{2(n-1)}$. Thus the net change in the virtual surplus will be:
\begin{equation*}
R'(q_i)\frac{k-1}{2(n-1)} - \frac{1}{2(n-1)}\sum_{j=i+1}^{k}R'(q_j)\geq 0
\end{equation*}
where the inequality follows since: $R'(\quant[i])\geq R'(\quant[j])$ for all $j\in[i+1,i+k]$. Thus it must be that each user is assigned a distinct badge.

The second part of the theorem is easy to see: The leaderboard mechanism falls into
the anonymous order-based framework of Chawla and Hartline \cite{Hartline2012}, and hence it follows that the mechanism will have a unique equilibrium which is symmetric and monotone. By this
fact, the bidders contributions will be decreasing in quantile and thereby the allocation implemented by the auction at the unique equilibrium is the same as the direct mechanism that
ranks users in decreasing order of quantiles.
\end{proof}

\begin{theorem}
If $\beta=0$, then the virtual surplus maximizing allocation assigns all users with quantile $q\leq \tqstar$ the highest badge and then assigns a distinct badge in decreasing order of quantile (increasing order of ability) to all users with quantile $q>\tqstar$. Such an allocation can be implemented at the unique equilibrium of a badge mechanism that assigns the top badge to
all users that pass a contribution threshold of $\theta=\val(\tqstar)\cdot \tqstar + \int_{\tqstar}^{1} \val(\quant)d\quant$ and then assigns a distinct badge in decreasing
order of contribution, to all users that don't pass the contribution threshold $\theta$. 
\end{theorem}
\begin{proof}
To argue the first part of the theorem, we argue that for any instantiation of user quantiles, the virtual surplus maximizing allocation is to assign the highest badge to all users with positive virtual ability and then order the remaining users in decreasing order of quantile. Similar to Lemma \ref{lem:optimal-surplus}, it is easy to see that the badge should be monotone non-decreasing in the virtual ability, since if $R'(\quant[i])>R'(\quant[j])$ and $r_i<r_j$ then we can increase virtual surplus by swapping the badge of user $i$ and user $j$. 

First it is easy to see that all users with non-negative virtual ability are assigned in the optimal mechanism, with no loss of generality, to the highest badge: by such an assignment each virtual ability is multiplied by $\stat{0}$, which is the highest possible status value that could be assigned to a user. Thus by not assigning a user with a positive virtual ability the highest badge, we are only multiplying his positive virtual ability by a smaller number and this decrease is not counterbalanced by some increase in another users status value. 
Thus in the optimal allocation, all users with positive virtual ability are assigned the highest badge. 
 
Now consider a set of virtual abilities $R'(\quant[i])\geq 0 > R'(\quant[i+1])\geq \ldots \geq R'(\quant[i+k])$ with $r_i=\ldots = r_{i+k}$ then by discriminating user $i$ to a higher badge then we can argue that the virtual surplus will increase:  More concretely, for any $j> i$ we can set $r_j'=r_j+1$. The status value of user $i$ and any user above $i$ remains unchanged, while the status value of user $j\in [i+1,i+k]$, will decrease by a factor of $\frac{1}{n+1}$. Since those users have negative virtual ability, the virtual surplus will increase. Thus it must be that all negative virtual ability users are assigned strictly lower badge than the positive virtual ability users. 

Next we argue that all negative virtual ability users must be strictly ranked. Suppose that for some set of virtual abilities $0 > R'(\quant[i])\geq \ldots \geq R'(\quant[i+k])$, we assign $r_i=\ldots = r_{i+k}$. Then by assigning $r_j=r_j+1$, to all users $j>i$, then the status value of user $i$ will remain unchanged, while the status value of all $j>i$, will 
decrease by $\frac{1}{n-1}$. Since all those users have negative virtual ability, the virtual surplus will increase. This completes the first part of the theorem.

To show the second part of the theorem we first argue that the following mechanism has a unique equilibrium: solicit contributions from the users. If a user's contribution surpasses a contribution threshold of $\theta$, then he is assigned the highest badge $n$. All users whose contribution is doesn't pass threshold $\theta$ are assigned a distinct rank in decreasing order of
contribution, starting from rank $n-1$ (breaking ties uniformly at random).

Given a contribution profile $b$, let $w(b)=|j: b_j<b_i|$ and $e(b)=|j: b_j=b_i|$, be the number of bidders that have strictly higher and and equal bid, correspondingly. Observe that the status allocation of a user in the above mechanism as a function of his bid is:
\begin{equation}
x_i(b) = \begin{cases}
1 & \text{if $b_i\geq \theta$}\\
1-\frac{w(b)}{n-1}-\frac{1}{2}\frac{e(b)-1}{n-1} & \text{if $b_i\leq \theta$}
\end{cases}
\end{equation} 
The second case holds, since a user will be uniformly at random ordered among users of equal rank and hence in expectation half of them will be ranked above him. Thus observe that the status allocation of a user is only a function of his bid $b_i$,
and of the relative rank $(w(b),e(b))$ of his bid among other bids. This makes the auction fall exactly into the badge of anonymous order-based auctions studied by Chawla and Hartline \cite{Hartline2012} and therefore, in the i.i.d. ability setting it will have a unique equilibrium, which will be symmetric and monotone.

Thus it suffices to give a specific setting of $\theta$, together with a symmetric bid equilibrium, that will implement the virtual surplus maximizing allocation. Assuming that the mechanism implements the virtual surplus maximizing badge allocation, we know by the equilibrium characterization that the bid of each user, will be:
\begin{equation}
b(\quant)=\begin{cases}
\val(\tqstar)\cdot \tqstar + \int_{\tqstar}^{1} \val(z)dz & \text{if $q\leq \tqstar$}\\
\int_{\quant}^{1} \val(z)dz & \text{if $\quant>\tqstar$}
\end{cases}
\end{equation}
Observe that the equilibrium bid has a discontinuity at $\tqstar$ and specifically, it jumps by $\val(\tqstar)\cdot \tqstar$. This is due to the extra status that a player gains from passing the top contribution threshold. For this bidding equilibrium to actually implement the claimed optimal direct allocation it must then be that:
\begin{equation}
\theta = \val(\tqstar)\cdot \tqstar + \int_{\tqstar}^{1} \val(z)dz 
\end{equation}
Under such a contribution threshold for the top badge, the equilibrium described previously implements the virtual surplus maximizing allocation. Thus the bid function and the optimal interim allocation of status satisfy the conditions of Lemma \ref{lem:myerson-bne}. To conclude that they are actually an equilibrium we simply need to argue that no player wants to bid in the region of bids that are not spanned by $b(\cdot)$, which is the region of bids in between the jump at $\tqstar$. This is trivially true, since for any such bid the player would prefer to bid $\int_{\tqstar}^1 \val(z)dz$. Thus the latter pair of bid and interim allocation are an equilibrium and by uniqueness, the unique equilibrium.
\end{proof}

Last we note, that for other values of $\beta$, the ex-post virtual surplus maximizing allocation of badges depends on the ex-post instance of the quantile profile and thereby doesn't have an ex-ante ranking or contribution threshold interpretation.